\documentclass[a4paper,reqno]{amsart}

\usepackage{amsfonts}
\usepackage{amssymb}
\usepackage{graphicx}

\DeclareMathOperator{\minor}{minor}
\DeclareMathOperator{\const}{const}

\newtheorem{theorem}{Theorem}
\newtheorem{lemma}{Lemma}

\theoremstyle{remark}
\newtheorem{remark}{Remark}

\author{S.I. Bel'kov, I.G. Korepanov, E.V. Martyushev}
\title[A simple TQFT for manifolds with triangulated boundary]{A simple topological quantum field theory for manifolds with triangulated boundary}
\date{}

\begin{document}

\begin{abstract}
We construct a simple finite-dimensional topological quantum field theory for compact 3-manifolds with triangulated boundary.
\end{abstract}

\maketitle

{ \small
 \tableofcontents
 }

\section{Introduction}

\subsection{Atiyah's axioms for TQFT}

The concept of a topological quantum field theory (TQFT) has its physical and mathematical aspects. In theoretical physics, its role is mainly seen as a theory of quantum gravity, although such or similar theory may be relevant also for some other physical ``gauge'' fields. And mathematically, a TQFT deals with topological invariants of a tensor or similar nature attributed to manifolds with boundary. These invariants must satisfy some properties formalized as axioms in works of M.~Atiyah~\cite{atiyah,atiyah-book}.

The main idea in Atiyah's axioms is that, if manifolds are glued together over some components of their boundaries, a composition of the corresponding invariants, such as tensor convolution, is taken for the result of gluing. This comes naturally from physics and reflects, in a general form, properties of quantum scattering amplitudes.

Here we describe a simple finite-dimensional (involving no functional integrals) TQFT of such kind for compact 3-dimensional manifolds with boundary. Our theory deals with anticommuting (Grassmann) variables attributed to edges of a manifold triangulation. We note that this corresponds to a modification of Atiyah's axioms explicitly mentioned by himself\footnote{``the vector spaces\ \dots\ may be mod~2 graded with appropriate signs then inserted'' --- \cite[\S\,2]{atiyah}}.

\subsection{Pachner moves and manifold invariants}
\label{subsec:pachner_moves}

The topological invariants in our theory are calculated out of a given manifold triangulation. If the boundary of a manifold is empty, then, to ensure that some value is a topological invariant\footnote{which is in three dimensions the same as piecewise-linear invariant~\cite{PL_and_Top}}, it is enough to prove its invariance under \emph{Pachner moves}. Recall that there are four Pachner moves in three dimensions: $2\leftrightarrow 3$ and $1\leftrightarrow 4$, see, for instance,~\cite{pachner_moves}.

The most interesting is, however, the case of a manifold with boundary. A triangulation of such manifold induces then a triangulation of the boundary. Our invariants will be constructed for a \emph{given} boundary triangulation, i.e., they do not depend oh a manifold triangulation provided it induces the given fixed triangulation of the boundary. In this case, the transition between different triangulations of the interior is achieved by \emph{relative} Pachner moves --- moves not involving the boundary. This has been explained in detail in~\cite{dkm}; the specific sort of boundary dealt with in~\cite{dkm} (specially triangulated torus) plays practically no role for the reasoning, which is directly generalized to the case of a general boundary.

\subsection{Organization}

Below, in section~\ref{sec:pentagon} we present a simple solution to pentagon equation (an algebraic relation corresponding to Pachner move~$2\to 3$) built of anticommuting variables. This already provides a set of topological invariants in some simple cases. The general situation requires, however, a more profound approach, based on algebraic (chain) complexes. So we give first, in section~\ref{sec:complexes}, the direct description of these complexes with all formulas needed for calculations, and then, we explain in section~\ref{sec:macroscopic_complexes} the ideas behind these formulas.

The resulting invariants are defined in section~\ref{sec:inv}, and then we explain in section~\ref{sec:genfun} how they are united in a ``generating function'' of anticommuting variables.

As we stated already, our invariants are constructed for a given boundary triangulation. So, in section~\ref{sec:boundary_change} we provide formulas answering the natural question of how they are changed under a change of boundary triangulation. We also prove in this section a lemma showing in which exactly cases the simplest invariants of section~\ref{sec:pentagon} work and how they are related to our more general approach.

The next section~\ref{sec:gluing} is central for justifying the name ``TQFT'' for our theory: in it, we give the formula for composition of our generating functions under the gluing of manifolds over a component of their boundaries. As a simple application of this, we study invariants for connected sums of manifolds in section~\ref{sec:connected_sums}.

In section~\ref{sec:examples} we provide some example calculations. Finally, we discuss our results in section~\ref{sec:discussion}.

\section{Solution to pentagon equation with anticommuting variables}
\label{sec:pentagon}

\subsection{Grassmann algebras and Berezin integral}

Recall~\cite{B} that \emph{Grassmann algebra} over field~$\mathbb R$ or~$\mathbb C$ is an associative algebra with unity, having generators~$a_i$ and relations
$$
a_i a_j = -a_j a_i, \quad \textrm{including} \quad a_i^2 =0.
$$
Thus, any element of a Grassmann algebra is a polynomial of degree $\le 1$ in each~$a_i$.

The \emph{Berezin integral}~\cite{B} in a Grassmann algebra is defined by equalities
\begin{equation}
\int da_i =0, \quad \int a_i\, da_i =1, \quad \int gh\, da_i = g \int h\, da_i,
\label{integral_Berezina}
\end{equation}
if $g$ does not depend on~$a_i$ (that is, generator $a_i$ does not enter the expression for~$g$); multiple integral is understood as iterated one.

\subsection{Solution to pentagon equation}

Consider a tetrahedron with vertices $i_1,i_2,\allowbreak i_3,i_4$, and let also this order of vertices (taken up to even permutations) determine its orientation. We will call such oriented tetrahedron simply ``tetrahedron~$i_1i_2i_3i_4$''.

\emph{Pentagon equation} is the name used by us, in a slightly informal way, for any algebraic relation which can be said to correspond naturally to a Pachner move~$2\to 3$. If such quantities are put in correspondence to the simplices in its l.h.s.\ and r.h.s.\ that this relation holds true, we say that a solution to pentagon equation has been found.

We introduce a complex parameter~$\zeta_i$ for every vertex~$i$, called its ``coordinate''. These parameters are arbitrary, with the only condition that any two different vertices~$i\ne j$ have different coordinates~$\zeta_i\ne\zeta_j$. We will also use the notation
$$
\zeta_{ij} \stackrel{\rm def}{=} \zeta_i - \zeta_j.
$$
Then, we put in correspondence to any unoriented edge~$ij$ a Grassmann generator~$a_{ij}=a_{ji}$, and to an oriented tetrahedron~$i_1i_2i_3i_4$ --- its \emph{generating function}
\begin{multline}
\mathbf f_{i_1i_2i_3i_4} \stackrel{\rm def}{=} \zeta_{i_1i_2}\zeta_{i_3i_4}(a_{i_1i_2}+a_{i_3i_4}) - \zeta_{i_1i_3}\zeta_{i_2i_4}(a_{i_1i_3}+a_{i_2i_4})\\ + \zeta_{i_1i_4}\zeta_{i_2i_3}(a_{i_1i_4}+a_{i_2i_3})
\label{fbf}
\end{multline}
The reason for the name generating function will be seen in section~\ref{sec:genfun}. We could also write $\mathbf f_{i_1i_2i_3i_4} = \mathbf f_{i_1i_2i_3i_4}(a_{i_1i_2},a_{i_1i_3},a_{i_1i_4},a_{i_2i_3},a_{i_2i_4},a_{i_3i_4})$ to emphasize that $\mathbf f_{i_1i_2i_3i_4}$ depends on these Grassmann variables.

\begin{theorem}
\label{th:fffff}
The function $\mathbf f_{i_1i_2i_3i_4}$ defined by~\eqref{fbf} satisfies the following pentagon equation (dealing with two tetrahedra $1234$ and~$5123$ in its l.h.s.\ and three tetrahedra $1254$, $2354$ and~$3154$ in its r.h.s.):
\begin{equation}
\mathbf f_{1234} \mathbf f_{5123} = \frac{1}{\zeta_{45}^2} \int \mathbf f_{1254} \mathbf f_{2354} \mathbf f_{3154} \, da_{45} .
\label{fffff}
\end{equation}
\end{theorem}

\begin{proof}
Formula~\eqref{fffff} can be proven, e.g., by a computer calculation.
\end{proof}

\begin{remark}
\label{rem:45}
The special role of edge~$45$ in~\eqref{fffff}, manifested in the factor $1/\zeta_{45}^2$ and integration in~$da_{45}$, corresponds obviously to the fact that~$45$ is the only inner edge among the ten edges of the r.h.s.\ tetrahedra.
\end{remark}

\subsection{A tentative state-sum invariant and the need for renormalization}
\label{subsec:renorm}

If there is a triangulated oriented manifold~$M$ with boundary, then one can construct the following function of anticommuting variables~$a_{ij}$ living on \emph{boundary} edges (and parameters~$\zeta_i$ in vertices):
\begin{equation}
\frac{1}{\prod\nolimits'\zeta_{ij}^2}\, \int \prod \mathbf f_{klmn} \prod\nolimits' da_{ij},
\label{s}
\end{equation}
where each of the two dashed products goes over all \emph{inner} edges~$ij$, while the remaining product --- over all oriented tetrahedra~$klmn$. As no preferred order of functions~$\mathbf f_{klmn}$ or differentials~$da_{ij}$ is fixed, the expression~\eqref{s} is determined up to an overall sign. It is a quite obvious consequence from theorem~\ref{th:fffff} and remark~\ref{rem:45} that \eqref{s} is at least invariant under all Pachner moves~$2\leftrightarrow 3$ not changing the boundary.

It turns out that~\eqref{s} is already, in some cases, a working multicomponent (that is, incorporating many coefficients at various monomials in anticommuting variables) invariant. We will call it in this paper the \emph{state sum} for manifold~$M$; from a physical viewpoint, the anticommuting variables mean that this is a state sum of \emph{fermionic} nature. There turn out to be, however, two difficulties with direct application of~\eqref{s}:
\begin{itemize}
\item if the triangulation has at least one inner (not boundary) vertex, \eqref{s} yields zero,
\item if the boundary of a connected manifold has more than one connected component, \eqref{s} also yields zero,
\end{itemize}
as we will show in lemma~\ref{lemma:t_i}.

It turns out that the \emph{renormalization} of state sum~\eqref{s}, leading to richer results, is achieved by introducing new variables, united in an algebraic (chain) complex.

\section{Algebraic complexes: explicit formulas for calculations}
\label{sec:complexes}

We consider a three-dimensional compact oriented manifold~$M$ with boundary~$\partial M$. Let it also be connected; otherwise, the following constructions can be done for each of its components separately. Our aim is to present (below in section~\ref{sec:inv}) a set of invariants, constructed for the given \emph{boundary} triangulation and depending on complex variables~$\zeta_i$ assigned to each boundary vertex~$i$; every individual invariant from the set corresponds to an ordered set~$\mathcal D$ of ``marked'' boundary edges. We also assume the following technical condition: the number of triangulation vertices in any connected component of~$\partial M$ is~$\ge 4$, unless the contrary is stated explicitly.

In this section, we present the formulas defining our algebraic complexes in the explicit form: essentially, as a sequence of five matrices~$f_1,\dots,f_5$. These formulas are well suited for computer calculations, although their form can hardly explain how they were found and for what reason our sequence~\eqref{complex_formal} of vector spaces and linear mappings is indeed an algebraic complex. This is explained in the next section~\ref{sec:macroscopic_complexes}.

Our invariants come out from algebraic (chain) complexes of  the following form\footnote{Some algebraic complexes of such kind have been already written out in~\cite[formulas~(29), (32), (49)]{kkm}. The main new feature of our complex~\eqref{complex} is that it works also for multicomponent boundary, which is due to introducing new quantities --- boundary component sways.}:
\begin{equation}
0 \to \mathbb C^3 \stackrel{f_1}{\to} \mathbb C^{N'_0+3m} \stackrel{f_2}{\to} \mathbb C^{N_3} \stackrel{f_3}{\to} \mathbb C^{N'_0+N_3} \stackrel{f_4}{\to} \mathbb C^{2N'_0+3m} \stackrel{f_5}{\to} \mathbb C^3 \to 0.
\label{complex_formal}
\end{equation}
Here $N'_0$ is the number of \emph{inner} vertices in the triangulation; $N_3$ is the number of all tetrahedra; $m$~is the number of connected components in~$\partial M$. We consider each vector space in~\eqref{complex_formal} as consisting of column vectors of the height equal to the exponent at~$\mathbb C$; all vector spaces have thus natural \emph{distinguished bases} consisting of vectors with one coordinate unity and all other --- zero (e.g., basis in~$\mathbb C^3$ consists of $\left(\begin{smallmatrix}1\\0\\0\end{smallmatrix}\right)$, $\left(\begin{smallmatrix}0\\1\\0\end{smallmatrix}\right)$ and~$\left(\begin{smallmatrix}0\\0\\1\end{smallmatrix}\right)$). We define linear mappings~$f_1,\dots,f_5$ --- which we identify with their matrices --- as follows.

\subsubsection*{Matrix $f_1$}

We denote a typical vector in the first nonzero space, $\mathbb C^3$, as~$\left( \begin{smallmatrix} da\\ db\\ dc \end{smallmatrix} \right)$; here and below the differential sign~$d$ is due to the differential nature of our vectors explained below in section~\ref{sec:macroscopic_complexes}. A typical vector in the next space, $\mathbb C^{N'_0+3m}$, is a column consisting of differentials~$dz_i$ living in each inner triangulation vertex~$i$, and also subcolumns~$\left( \begin{smallmatrix} ds_{\kappa}^{(a)}\\ ds_{\kappa}^{(b)}\\ ds_{\kappa}^{(c)} \end{smallmatrix} \right)$ living on each connected component~$\kappa$ of~$\partial M$ --- we call such subcolumn (infinitesimal) \emph{sway} of component~$k$, see explanation in section~\ref{sec:macroscopic_complexes}. The action of matrix $f_1$ gives, by definition:
\begin{equation}
dz_i = \begin{pmatrix}2\zeta_i & 1 & -\zeta_i^2 \end{pmatrix} \begin{pmatrix} da\\ db\\ dc \end{pmatrix} ;
\qquad \begin{pmatrix} ds_{\kappa}^{(a)}\\ ds_{\kappa}^{(b)}\\ ds_{\kappa}^{(c)} \end{pmatrix} = \begin{pmatrix} da\\ db\\ dc \end{pmatrix} .
\label{f1}
\end{equation}
In other words, $f_3$ consists of submatrices $\begin{pmatrix}2\zeta_i & 1 & -\zeta_i^2 \end{pmatrix}$ and~$\left( \begin{smallmatrix} 1&0&0 \\ 0&1&0 \\ 0&0&1 \end{smallmatrix} \right)$.

\subsubsection*{Matrix $f_2$}

A typical vector in the next (third nonzero from the left in~\eqref{complex_formal}) space, $\mathbb C^{N_3}$, is a column consisting of differentials~$dy_{ijkl}$ living in each (oriented) tetrahedron~$ijkl$. If all vertices~$i,j,k,l$ are inner, the action of matrix~$f_2$ gives, by definition:
\begin{equation}
dy_{ijkl}=\begin{pmatrix} -\frac{1}{\zeta_{ij}\zeta_{ik}\zeta_{il}} & -\frac{1}{\zeta_{ij}\zeta_{jk}\zeta_{jl}} &
-\frac{1}{\zeta_{ik}\zeta_{jk}\zeta_{kl}} & -\frac{1}{\zeta_{il}\zeta_{jl}\zeta_{kl}} \end{pmatrix} \begin{pmatrix} dz_i\\ dz_j\\ dz_k\\ dz_l
\end{pmatrix} .
\label{f2_1}
\end{equation}
If some of the vertices~$i,j,k,l$ is/are boundary, formula~\eqref{f2_1} still holds, with every~$dz_m$ for a boundary vertex~$m$ belonging to boundary component~$\kappa$ (recall that $dz_m$ is \emph{absent} from the vector columns in~$\mathbb C^{N'_0+3m}$; it is just some auxiliary quantity) defined as follows:
\begin{equation}
dz_m \stackrel{\rm def}{=} \begin{pmatrix}2\zeta_m & 1 & -\zeta_m^2 \end{pmatrix} \begin{pmatrix} ds_{\kappa}^{(a)}\\ ds_{\kappa}^{(b)}\\ ds_{\kappa}^{(c)} \end{pmatrix} .
\label{f2_2}
\end{equation}

\begin{remark}
\label{rem:several_simplexes}
There may well be \emph{several} tetrahedra in the triangulation having the same vertices~$i,j,k,l$. In this case, each of them has, of course, its own quantity~$dy$, so, in practical calculations, we will have to use more complicated notations for tetrahedra than just~$ijkl$. We think, however, that when we focus on just one tetrahedron, like in formula~\eqref{f2_1}, our notations are perfectly justified.

The same will apply below to our notations like ``$ij$'' for edges.
\end{remark}

\subsubsection*{Matrix $f_3$}

A typical vector in the fourth nonzero space in~\eqref{complex_formal}, $\mathbb C^{N'_0+N_3}$, is a column consisting of differentials~$d\varphi_{ij}=d\varphi_{ji}$ for the set of edges~$ij$ including all inner edges --- we denote their number as~$N'_1$ --- and also a set~$\mathcal D$ of ``marked'' boundary edges. The total number~$N'_0+N_3$ of such edges is determined by the condition of vanishing of the \emph{Euler characteristics} (the alternated sum of dimensions of vector spaces) of complex~\eqref{complex_formal}. This can work due to the following lemma.

\begin{lemma}
\label{lemma:boundary_edges}
Let $N_i$ denote the number of $i$-dimensional simplexes in a triangulation of manifold~$M$, and $N'_i$ --- the number of inner (not lying entirely in the boundary) $i$-dimensional simplexes. Then
\begin{equation}
N'_1 \le N'_0+N_3 \le N_1.
\label{enough_edges}
\end{equation}
Moreover, if $\partial M$ is nonempty, both inequalities~\eqref{enough_edges} become strict, while for the empty $\partial M$ they turn into equalities.
\end{lemma}

\begin{proof}
Consider first some closed three-dimensional triangulated manifold~$\tilde M$ with $\tilde N_i$ the number of simplexes of dimension~$i$. As is known, its Euler characteristics $\tilde N_0-\tilde N_1+\tilde N_3=0$ (here the l.h.s.\ can be written in this form because $\tilde N_2=2\tilde N_3$). We apply this to $\tilde M$ being the doubled~$M$ (i.e., two oppositely oriented copies of~$M$ glued naturally over their whole boundaries):
$$
2N'_0+n_0-2N'_1-n_1+2N_3=0,
$$
where $n_0=N_0-N'_0$ and~$n_1=N_1-N'_1$ are the numbers of vertices and edges in the boundary. Hence, $N'_0+N_3-N'_1=\frac{1}{2}(n_1-n_0)$, and \eqref{enough_edges} reduces to
\begin{equation}
-n_0-n_1 \le 0 \le n_1-n_0.
\label{nsmall}
\end{equation}

The first inequality~\eqref{nsmall} is evident, as well as all lemma statement concerning it. To prove the second inequality~\eqref{nsmall}, we note that the Euler characteristics of~$\partial M$ (which is a closed triangulated two-dimensional manifold) can be written, without using the number of two-dimensional cells, as $\chi_{\partial M}=n_0-\frac{1}{3}n_1$, i.e., $n_1-n_0=2n_0-3\chi_{\partial M}$. It remains to recall that the contribution of each boundary component in $n_0$, as we agreed in the beginning of this section, is not less than~$4$, while in $\chi_{\partial M}$ --- not greater than~$2$.
\end{proof}

The action of matrix~$f_3$ gives, by definition:
\begin{equation}
d\varphi_{ij}=\zeta_{ij} \sum_{\mathrm{edges\ }kl} \zeta_{kl} \,dy_{ijkl},
\label{f3}
\end{equation}
where ``edges $kl$'' are those edges belonging to the link of $ij$ which are either inner or belong to the set~$\mathcal D$; the order of vertices $ijkl$ must correspond to the orientation of this tetrahedron induced by the orientation of~$M$.

\subsubsection*{Matrix $f_4$}

A typical vector in the fourth nonzero space in~\eqref{complex_formal}, $\mathbb C^{2N'_0+3m}$, is a column consisting of differentials~$d\alpha_i$ and~$d\beta_i$ for each inner vertex~$i$, and also subcolumns $\left( \begin{smallmatrix} dt_{\kappa}^{(a)}\\ dt_{\kappa}^{(b)}\\ dt_{\kappa}^{(c)} \end{smallmatrix} \right)$ for each boundary component~$\kappa$; we call these subcolumns \emph{conjugate sways}. The action of matrix~$f_4$ gives for $d\alpha_i$ and~$d\beta_i$, by definition:
\begin{equation}
\begin{pmatrix} d\alpha_i \\ d\beta_i \end{pmatrix} = \sum_{\textrm{edges }ij} \begin{pmatrix} 1 \\ 1/\zeta_{ij} \end{pmatrix} d\varphi_{ij} ,
\label{f4_1}
\end{equation}
where the sum is taken over all edges~$ij$ starting at~$i$.

We also \emph{define} the differentials~$d\alpha_i$ and~$d\beta_i$ for each \emph{boundary} vertex~$i$ --- just as auxiliary quantities entering the following formula~\eqref{f4_2} --- by the same formula~\eqref{f4_1}, where the sum is now taken over all \emph{inner} edges~$ij$ starting at~$i$. The action of matrix~$f_4$ gives for the conjugate sways, by definition:
\begin{equation}
\begin{pmatrix} dt_{\kappa}^{(a)}\\ dt_{\kappa}^{(b)}\\ dt_{\kappa}^{(c)} \end{pmatrix} = \sum_i \begin{pmatrix} -1&2\zeta_i \\ 0&1 \\ \zeta_i&-\zeta_i^2 \end{pmatrix} \begin{pmatrix} d\alpha_i \\ d\beta_i \end{pmatrix} ,
\label{f4_2}
\end{equation}
where the sum is taken over all vertices~$i$ belonging to boundary component~$\kappa$.

\subsubsection*{Matrix $f_5$}

We write a typical vector in the last nonzero space in~\eqref{complex_formal}, $\mathbb C^3$, as $ \left( \begin{smallmatrix} da^*\\db^*\\dc^* \end{smallmatrix} \right) $. The action of matrix~$f_5$ gives, by definition:
\begin{equation}
\begin{pmatrix}da^*\\db^*\\dc^*\end{pmatrix} = \sum_i \begin{pmatrix} -1&2\zeta_i \\ 0&1 \\ \zeta_i&-\zeta_i^2 \end{pmatrix} \begin{pmatrix} d\alpha_i \\ d\beta_i \end{pmatrix} + \sum_{\kappa} \begin{pmatrix} dt_{\kappa}^{(a)}\\ dt_{\kappa}^{(b)}\\ dt_{\kappa}^{(c)} \end{pmatrix} ,
\label{f5}
\end{equation}
where the first sum in the r.h.s.\ is taken over all inner vertices~$i$, while the second --- over all boundary components~$\kappa$.

\begin{theorem}
\label{th:complex_formal}
The sequence \eqref{complex_formal} is indeed an algebraic complex, i.e.:
\begin{equation}
f_2 \circ f_1 = 0, \quad f_3 \circ f_2 = 0, \quad f_4 \circ f_3 = 0, \quad f_5 \circ f_4 = 0.
\label{ff=0}
\end{equation}
\end{theorem}

\begin{proof}
The equalities \eqref{ff=0} can be proved using directly the definitions of $f_1,\dots,f_5$ given in this section.

We do not give here the details of these direct calculations, because a different proof of theorem~\ref{th:complex_formal} will follow from our further reasoning, see remarks \ref{rem:left} and~\ref{rem:right}.
\end{proof}

\section{Algebraic complexes: the mathematical origins}
\label{sec:macroscopic_complexes}

The presented direct proof of theorem~\ref{th:complex_formal} does not make clear the mathematical reasons ensuring that \eqref{complex_formal} is a complex. To understand these reasons is also desirable for proving theorem~\ref{th:minors_nonzero} below in section~\ref{sec:inv}. So, this section is devoted to explaining the mathematical origins of complex~\eqref{complex_formal}. We mainly follow sections 2 and~3 from~\cite{kkm}, modifying them in such way as to include the case of a multi-component boundary~$\partial M$.

\subsection{The left-hand half of the complex}
\label{subsec:left}

Recall that we are considering a three-dimensional closed oriented connected manifold~$M$ with boundary~$\partial M$. We attach a complex number~$\zeta_i$ to every vertex~$i$ of its given triangulation; $\zeta_i$ will be called, from now on, the unperturbed, or initial, coordinate\footnote{as opposed to ``perturbed'' coordinates~$z_i$ below} of vertex~$i$. Recall also that $N_i$ is the number of $i$-dimensional simplexes in the triangulation, and $m$ is the number of connected components in~$\partial M$.

We are going to define the following chain of spaces and (nonlinear) mappings:
\begin{multline}
0 \longrightarrow \mathrm{PSL}(2,\mathbb C) \stackrel{\textstyle F_1}{\longrightarrow}
\left( \begin{smallmatrix} \text{inner vertex}\\ \text{coordinates}\\ z \end{smallmatrix}\right) \oplus \left( \begin{smallmatrix} \text{boundary}\\ \text{component}\\ \text{sways }s \end{smallmatrix}\right) \\ \stackrel{\textstyle F_2}{\longrightarrow}
\left( \begin{smallmatrix} \text{triples}\\ x,\,1-1/x,\,1/(1-x) \\ \text{in tetrahedra}\end{smallmatrix}\right) \stackrel{\textstyle F_3}{\longrightarrow}
\left( \begin{smallmatrix} \text{total}\\ \text{angles }\omega\\ \text{around edges} \end{smallmatrix}\right).
\label{crmacro}
\end{multline}

The leftmost arrow sends, by definition, the zero into the unit of group $\mathrm{PSL}(2,\mathbb C)$.

Mapping~$F_1$ sends an element of group~$\mathrm{PSL}(2,\mathbb C)$ represented by matrix $\left( \begin{smallmatrix} \alpha&\beta\\
\gamma&\delta \end{smallmatrix}\right)$ into the direct sum of two column vectors. The first of them is of height~$N'_0$ and consists of complex numbers~$z_i$ called ``perturbed coordinates'' of all \emph{inner} vertices~$i$. By definition, $F_1$ builds from the mentioned matrix the numbers
\begin{equation}
z_i=\frac{\alpha\zeta_i+\beta}{\gamma\zeta_i+\delta}.
\label{f1_1_macro}
\end{equation}

The second column vector in the mentioned direct sum is of height~$m$, and each of its entries is just a copy of the same group~$\mathrm{PSL}(2,\mathbb C)$ which we put in correspondence to each boundary component and call its \emph{sway}. By definition, each of these $m$~components of~$F_1$ takes any element of~$\mathrm{PSL}(2,\mathbb C)$ into itself (thus resulting in $m$ identical sways of boundary components).

\begin{remark}
\label{rem:sways}
By ``sway'' we mean, speaking less formally, a motion of the whole boundary component as a rigid body, in contrast with inner vertices which are allowed to move independently, as will be seen in the coming definition of mapping~$F_2$. This applies as well to the sways~$t^*$ below in subsection~\ref{subsec:right}.
\end{remark}

The next mapping~$F_2$ sends the pair (column vector of $N'_0$ \emph{arbitrary} values~$z_i$, column vector of $m$ arbitrary elements of~$\mathrm{PSL}(2,\mathbb C)$) into the column vector of height~$N_3$, whose each entry corresponds to a tetrahedron in the triangulation and is described as follows. First, we introduce the perturbed coordinates of the \emph{boundary} vertices --- just as auxiliary quantities, \emph{not} entering directly our sequence~\eqref{crmacro}. By definition, they are given by the same formula~\eqref{f1_1_macro} as for inner vertices.

Let now there be a tetrahedron~$ijkl$, whose orientation (given by this order of its vertices) corresponds to the given orientation of~$M$. The entry of the mentioned vector, corresponding\footnote{Recall that, according to remark~\ref{rem:several_simplexes}, the situation where there are several tetrahedra having the same vertices $i,j,k,l$ is perfectly acceptable; we will just have to use more complicated notations to distinguish them; the same applies to edges denoted like~``$ij$''.} to tetrahedron~$ijkl$, consists of three complex values corresponding to its six \emph{unoriented} edges and related as follows:
\begin{itemize}
\item the same value corresponds to any of two opposite edges: if $x$ corresponds to edge~$ik$, it also corresponds to edge~$jl$;
\item if $x$ corresponds to edges $ik$ and~$jl$, then the first of the values
\begin{equation}
1-\frac{1}{x}\,,\quad \frac{1}{1-x}
\label{2more}
\end{equation}
corresponds to any of the edges $il$ and $jk$, while the second --- to the edges $ij$ and~$kl$.
\end{itemize}
By definition, the~$x$ obtained by applying $F_2$ to given~$z$'s is the cross-ratio
\begin{equation}
x=\frac{z_{ij}z_{kl}}{z_{il}z_{kj}},
\label{cr}
\end{equation}
where
\begin{equation}
z_{ij}=z_i-z_j
\label{zij}
\end{equation}
(and $z_i$ for inner and boundary vertices are on equal footing in~\eqref{cr}). One can check that expressions \eqref{2more} are in accordance with how the cross-ratio~\eqref{cr} transforms under permutations of vertices.

Finally, to describe mapping~$F_3$, we choose a set~$\mathcal D$ of ``marked'' boundary edges of such cardinality~$\# \mathcal D$ that
$$
N'_1+\# \mathcal D = N'_0+N_3
$$
in the same way as in section~\ref{sec:complexes}; recall that this can be done due to lemma~\ref{lemma:boundary_edges}. Mapping~$F_3$ sends a column vector of height~$N_3$ consisting of triples~$\bigl(x,\allowbreak\;1-1/x,\allowbreak\;1/(1-x)\bigr)$ into a column vector of complex numbers~$\omega_{ij}$ of height~$N'_1+\# \mathcal D$, where $ij$ denotes an edge joining vertices $i$ and~$j$. Consider the \emph{star} of~$ij$; it consists of all tetrahedra having $ij$ as an edge. By definition, $F_3$ yields
\begin{equation}
\omega_{ij}=\prod x,
\label{oij}
\end{equation}
where all values $x$ in the product correspond to all tetrahedra in the star of~$ij$ and to the edge~$ij$ in each such tetrahedron. We call $\omega_{ij}$ obtained according to formula~\eqref{oij} \emph{total angle} around edge~$ij$.

For inner edges, the total angle is of course the same as the ``deficit angle'' of paper~\cite{kkm}.

\begin{theorem}
\label{th-l-glob}
The composition of any two successive arrows in~(\ref{crmacro}) is a constant mapping.
\end{theorem}

\begin{proof}
To show that $F_2\circ F_1=\const$, it is enough to say that the cross-ratio of four complex numbers is invariant under the action of the same element of~$\mathrm{PSL}(2,\mathbb C)$ on all of them.

To show that $F_3\circ F_2=\const$, we denote the successive vertices in the \emph{link} of edge~$ij$ as $1,\dots,r$, so that the oriented tetrahedra around~$ij$ are $ij12$, $\ldots$, $ij(r-1)r$, $ijr1$ in the case if $ij$ is an inner edge or just $ij12$, $\ldots$, $ij(r-1)r$ in the case if $ij$ is a boundary edge. Then the product~(\ref{oij}) of values~(\ref{cr}) is
$$
\omega_{ij} = \frac{z_{i2}z_{j1}}{z_{j2}z_{i1}} \cdots \frac{z_{ir}z_{j(r-1)}}{z_{jr}z_{i(r-1)}} \frac{z_{i1}z_{jr}}{z_{j1}z_{ir}} = 1
$$
for the inner $ij$ or
\begin{equation}
\omega_{ij} = \frac{z_{i2}z_{j1}}{z_{j2}z_{i1}} \cdots \frac{z_{ir}z_{j(r-1)}}{z_{jr}z_{i(r-1)}} = \frac{z_{j1}z_{ir}}{z_{i1}z_{jr}} 
\label{ij1r}
\end{equation}
for the boundary $ij$. The ``inner'' case is obvious, while in the ``boundary'' case it remains to note that all vertices $i,j,1,r$ lie in the boundary, so neither changes of inner~$z_k$ nor action of~$\mathrm{PSL}(2,\mathbb C)$ due to boundary sways can affect the (rightmost) cross-ratio~\eqref{ij1r}.
\end{proof}

We sometimes call the chain~(\ref{crmacro}) a ``macroscopic'' complex, in contrast to its differential, or ``microscopic'' version which we are going to produce from it and which will coincide with the left-hand half of~\eqref{complex_formal} (including the arrow~$f_3$). Roughly speaking, it will consist of differentials of mappings $F_1$, $F_2$ and~$F_3$.

This makes no difficulty when taking the differential 
$$
f_1=dF_1\colon\; \mathfrak{psl}(2,\mathbb C)\to (dz) \oplus (ds) ,
$$
where $\mathfrak{psl}(2,\mathbb C)$ is the Lie algebra, $(dz)$ denotes the vector space of column vectors of differentials of quantities~$z_i$, (more formally, $(dz)$ is just a vector space over~$\mathbb C$ whose basis consists of all the vertices of triangulation) and $(ds)$ denotes the the vector space which is the direct sum of $m$ copies of~$\mathfrak{psl}(2,\mathbb C)$. To be exact, we choose the natural basis of three matrices 
\begin{equation}
\begin{pmatrix} 1&0 \\ 0&-1 \end{pmatrix}, \quad
\begin{pmatrix} 0&1 \\ 0&0 \end{pmatrix} \quad \text{and} \quad
\begin{pmatrix} 0&0 \\ 1&0 \end{pmatrix}
\label{basis_psl}
\end{equation}
in~$\mathfrak{psl}(2,\mathbb C)$, denote the coordinates with respect to it as~$da,db,dc$ in the algebra to the left of arrow~$f_1$ and $ds_k^{(a)},ds_k^{(b)},ds_k^{(c)}$ in the sways of $k$th boundary component, and then a simple differentiation gives the already written formula~\eqref{f1} for~$f_1$.

For the next mapping, we would like to produce just one symmetric differential out of three ``macroscopic'' quantities (\ref{cr}) and~(\ref{2more}), namely
\begin{equation}
dy_{ijkl}=\frac{d\ln x}{\zeta_{ik}\zeta_{lj}}=\frac{d\ln (1-\frac{1}{x})}{\zeta_{il}\zeta_{jk}}=\frac{d\ln \frac{1}{1-x}}{\zeta_{ij}\zeta_{kl}}.
\label{dy}
\end{equation}
Our ``microscopic'' mapping 
$$
f_2\colon\; (dz)\oplus (ds) \to (dy)
$$
is defined by differentiating formula~\eqref{cr}; here $(dy)$ is the space of column vectors whose coordinates are~$dy_{ijkl}$ for all tetrahedra~$ijkl$ in the triangulation (more formally --- the vector space over~$\mathbb C$ whose basis consists of all the tetrahedra). The formulas for~$f_2$ are the already written formulas \eqref{f2_1} and~\eqref{f2_2}.

Finally, we introduce variables $\varphi_i = \ln \omega_i$ in our definition of ``microscopic'' mapping
$$
f_3\colon\; (dy)\to (d\varphi),
$$
where $(d\varphi)$ is again the obvious vector space, whose basis vectors are inner edges and edges from set~$\mathcal D$. The differential of~$F_3$ gives, in terms of these variables, the already written formula~\eqref{f3}.

Hence, our resulting sequence of vector spaces and linear mappings is:
\begin{equation}
0\longrightarrow \mathfrak{psl}(2,\mathbb C) \stackrel{f_1}{\longrightarrow} (dz)\oplus (ds) \stackrel{f_2}{\longrightarrow} (dy)
\stackrel{f_3}{\longrightarrow} (d\varphi)
\label{crmicro}
\end{equation}

\begin{remark}
\label{rem:left}
We have thus obtained a different proof of one-half of theorem~\ref{th:complex_formal}, reflecting really the ideas behind it. Indeed, the equalities $f_3\circ f_2=0$ and $f_2\circ f_1=0$ follow immediately by differentiation from theorem~\ref{th-l-glob}.
\end{remark}

\subsection{The right-hand half of the complex}
\label{subsec:right}

We define also one more ``macroscopic'' sequence of spaces and (nonlinear) mappings:
\begin{multline}
0 \longrightarrow \mathrm{SO}(3,\mathbb C) \stackrel{\textstyle G_1}{\longrightarrow}
\left( \begin{smallmatrix} \text{isotropic}\\ \text{vectors}\\ \text{in inner vertices} \end{smallmatrix}\right) \oplus \left( \begin{smallmatrix} \text{boundary}\\ \text{component}\\ \text{sways }t^* \end{smallmatrix}\right) \\
\stackrel{\textstyle G_2}{\longrightarrow}
\left( \begin{smallmatrix} \text{squared}\\ \text{edge}\\ \text{lengths} \end{smallmatrix}\right) \stackrel{\textstyle G_3}{\longrightarrow}
\left( \begin{smallmatrix} \text{discrepancies}\\ \Omega \\ \text{in tetrahedra}\end{smallmatrix}\right) .
\label{crmacroconj}
\end{multline}

Here are the details. The first arrow just maps the zero into the unity of the group~$\mathrm{SO}(3,\mathbb C)$. Note that this group is isomorphic to $\mathrm{PSL}(2,\mathbb C)$ with which we were dealing in subsection~\ref{subsec:left}. 

To move further, we have to consider a complex Euclidean space of column vectors of height~$3$ with the scalar product given by the matrix
\begin{equation}
\begin{pmatrix} 0&0&-1\\0&2&0\\-1&0&0 \end{pmatrix}.
\label{matrixEuclid}
\end{equation}
We realize the group $\mathrm{SO}(3,\mathbb C)$ as the group of matrices representing linear transformations of this space preserving the scalar product~(\ref{matrixEuclid}).

This time, we associate \emph{two} complex parameters with each vertex~$i$ of our manifold triangulation: $\zeta_i$ which is the same as in subsection~\ref{subsec:left}, and a new parameter called~$\varkappa_i$. These parameterize the following ``initial'', or unperturbed, \emph{isotropic vectors}:
\begin{equation}
\vec e_i^{\textrm{ initial}} = \begin{pmatrix} \varkappa_i \zeta_i^2\\\varkappa_i \zeta_i\\ \varkappa_i \end{pmatrix}.
\label{isotropic-initial}
\end{equation}
The space called ``$\left( \begin{smallmatrix} \text{isotropic}\\ \text{vectors}\\ \text{in inner vertices} \end{smallmatrix}\right)$'' in~(\ref{crmacroconj}) consists of isotropic vectors~$\vec e_i$ in all inner vertices~$i$ of the form~(\ref{isotropic-initial}), but with all $\zeta_i$ and~$\varkappa_i$ replaced by arbitrary complex values $z_i$ and~$h_i$:
\begin{equation}
\vec e_i=\begin{pmatrix} h_i z_i^2\\h_i z_i\\h_i \end{pmatrix}
\label{ei}
\end{equation}
As for the space ``$\left( \begin{smallmatrix} \text{boundary}\\ \text{component}\\ \text{sways }t^* \end{smallmatrix}\right)$'', it consists of $m$ copies of the same group~$\mathrm{SO}(3,\mathbb C)$.

By definition, our mapping~$G_1$ builds the following vectors~\eqref{ei}, for all inner vertices~$i$, out of an element $T\in\mathrm{SO}(3,\mathbb C)$:
\begin{equation}
G_1\colon\quad T \mapsto \{\textrm{vectors }\vec e_i = T \vec e_i^{\textrm{ initial}} \textrm{ for all }i\},
\label{G1}
\end{equation}
and also gives $m$ identical boundary component sways\footnote{The star in our notation $t^*$ and other notations below reflects the ``conjugation'' which will be done soon with the microscopic version of complex~\eqref{crmacroconj}.}: $t_{\kappa}^*=T$, $\kappa=1,\dots,m$.

The next space called ``$\left( \begin{smallmatrix} \text{squared}\\ \text{edge}\\ \text{lengths} \end{smallmatrix}\right)$'' in~(\ref{crmacroconj}) consists of complex numbers living on all inner edges and boundary edges in the set~$\mathcal D$. We assume that our isotropic vectors come out of the origin of coordinates. The map $G_2$ produces then for edge~$ij$, by definition, the squared distance~$L_{ij}$ between the ends of $\vec e_i$ and~$\vec e_j$. The sways~$t^*$ play here their usual role: if $i$ (or/and $j$) belongs to boundary component~$\kappa$, the ``perturbed'' vector~\eqref{ei} is used for it also, calculated according to
$$
\vec e_i = t_{\kappa}^* \vec e_i^{\textrm{ initial}} .
$$

Note the following relation between $L_{ij}$ and the scalar product:
\begin{equation}
L_{ij}=-2\vec e_i \vec e_j.
\label{Lij}
\end{equation}

Finally, our space ``$\left( \begin{smallmatrix} \text{discrepancies}\\ \Omega \\ \text{in tetrahedra}\end{smallmatrix}\right)$'' consists of complex numbers~$\Omega_{ijkl}$ put in correspondence to all tetrahedra~$ijkl$. By definition, the $\Omega$'s produced by~$G_3$ from the given $L$'s are the following determinants:
\begin{equation}
\Omega_{ijkl}=\left| \begin{matrix}
0&L_{ij}&L_{ik}&L_{il}\\
L_{ji}&0&L_{jk}&L_{jl}\\
L_{ki}&L_{kj}&0&L_{kl}\\
L_{li}&L_{lj}&L_{lk}&0
\end{matrix} \right|,
\end{equation}
where of course $L_{ij}=L_{ji}$ and so on. Here $L_{ij}$ is regarded as an independent complex variable if the edge~$ij$ is either inner of in the set~$\mathcal D$; otherwise, $L_{ij}$ is a constant, namely the distance between the ends of corresponding unperturbed vectors.

\begin{theorem}
\label{th-G-global}
The composition of any two successive arrows in~(\ref{crmacroconj}) is a constant mapping.
\end{theorem}

\begin{proof}
The relation $G_2\circ G_1 =\const$ holds simply because distances are invariant under the action of $\mathrm{SO}(3,\mathbb C)$.

The relation $G_3\circ G_2 =\mathrm{const}\; (=0)$ holds because $\Omega$ vanishes when the $L$'s are produced from \emph{three-dimensional} vectors according to~(\ref{Lij}).
\end{proof}

Now we pass on to ``microscopic'' values similarly to subsection~\ref{subsec:left}: we produce linear mappings $g_1$, $g_2$ and~$g_3$ as differentials $dG_1$, $dG_2$ and~$dG_3$ multiplied by some simple factors.

We choose the basis of three following matrices in the Lie algebra~$\mathfrak{so}(3,\mathbb C)$:
\begin{equation}
A=\begin{pmatrix}2&0&0\\0&0&0\\0&0&-2\end{pmatrix},\quad B=\begin{pmatrix}0&2&0\\0&0&1\\0&0&0\end{pmatrix},\quad
C=\begin{pmatrix}0&0&0\\1&0&0\\0&2&0\end{pmatrix}.
\label{basis_so}
\end{equation}

Let $da^*,db^*,dc^*$ be infinitesimal numbers; we also denote
\begin{equation}
d\alpha_i^*=\frac{dh_i}{2\varkappa_i},\quad d\beta_i^*=dz_i.
\label{alpha*beta*}
\end{equation}
If we calculate the change of $h_i$ and~$z_i$ under the action of matrix $da^*A+db^*B+dc^*C$ on vector $\vec e_i$~(\ref{ei}) and then substitute the initial values $h_i=\varkappa_i$ and $z_i=\zeta_i$ into the resulting Jacobian matrix, we get, taking also \eqref{alpha*beta*} into account:
\begin{equation}
\begin{pmatrix}d\alpha_i^*\\d\beta_i^*\end{pmatrix}=\begin{pmatrix}-1&0&\zeta_i\\
2\zeta_i&1&-\zeta_i^2\end{pmatrix} \begin{pmatrix}da^*\\db^*\\dc^*\end{pmatrix}.
\label{g1_1}
\end{equation}
By definition, linear mapping~$g_1$ sends a vector column $\left( \begin{smallmatrix}da^*\\db^*\\dc^*\end{smallmatrix} \right)$ into the set of differentials~\eqref{g1_1} for all inner vertices~$i$ and to the columns
\begin{equation}
\begin{pmatrix}dt_{\kappa}^{(a)*} \\ dt_{\kappa}^{(b)*} \\ dt_{\kappa}^{(c)*} \end{pmatrix} = \begin{pmatrix}da^*\\db^*\\dc^*\end{pmatrix}
\label{g1_2}
\end{equation}
for each boundary component~$\kappa$.

Next, we introduce ``normalized'' squared edge lengths in the following way:
\begin{equation*}
\varphi_{ij}^*=\frac{L_{ij}}{4 \varkappa_i \varkappa_j (\zeta_i-\zeta_j)^2}.
\end{equation*}
Thus, when $\varphi_{ij}^*$ is obtained according to $G_2$, it is
\begin{equation}
\varphi_{ij}^*=\frac{1}{2} \frac{h_i h_j (z_i-z_j)^2}{\varkappa_i \varkappa_j (\zeta_i-\zeta_j)^2}.
\label{eq:phiij}
\end{equation}
This yields
\begin{equation}
\frac{\partial \varphi_{ij}^*}{\partial \alpha_i^*}=1,\quad \frac{\partial \varphi_{ij}^*}{\partial \beta_i^*}=\frac{1}{\zeta_i-\zeta_j}.
\label{g2_1}
\end{equation}
By definition, formula~(\ref{g2_1}) gives matrix elements for linear mapping~$g_2$, together with the following analogue of formula~\eqref{g1_1} which must be used for calculating the differentials $d\alpha_i^*$ and~$d\beta_i^*$ for \emph{boundary} vertices:
\begin{equation}
\begin{pmatrix}d\alpha_i^*\\d\beta_i^*\end{pmatrix}=\begin{pmatrix}-1&0&\zeta_i\\
2\zeta_i&1&-\zeta_i^2\end{pmatrix} \begin{pmatrix}dt_{\kappa}^{(a)*} \\ dt_{\kappa}^{(b)*} \\ dt_{\kappa}^{(c)*} \end{pmatrix}.
\label{g2_2}
\end{equation}

Finally, if $\Omega_{ijkl}$ is obtained according to~$G_3$ and we calculate the derivative $\partial \Omega_{ijkl}/\partial \varphi_{ij}^*$ at the
point where $L_{ij}=-2\vec e_i\vec e_j =2\varkappa_i\varkappa_i (\zeta_i-\zeta_j)^2$ and similarly for $L$'s with other indices, we get
\begin{equation*}
\frac{\partial \Omega_{ijkl}}{\partial \varphi_{ij}^*}=-128 (\zeta_i-\zeta_j)(\zeta_k-\zeta_l) \prod_{r<s} (\zeta_r-\zeta_s),
\end{equation*}
where in the product both $r$ and $s$ take values $i,j,k,l$, and ``$<$'' in ``$r<s$'' means just the alphabetic order. This suggests us to denote
\begin{equation*}
dy_{ijkl}^*=-\frac{d\Omega_{ijkl}}{128\prod_{r<s} (\zeta_r-\zeta_s)},
\end{equation*}
which yields
\begin{equation}
\frac{\partial y_{ijkl}^*}{\partial \varphi_{ij}^*} = \frac{1}{\zeta_{ij}\zeta_{kl}}.
\label{yphi}
\end{equation}
By definition, \eqref{yphi} gives matrix elements for linear mapping~$g_3$.

Hence, the resulting ``microscopic'' sequence is
\begin{equation}
0\longrightarrow \mathfrak{so}(3,\mathbb C) \stackrel{g_1}{\longrightarrow} (d\alpha^*)\oplus(d\beta^*)\oplus(dt^*) \stackrel{g_2}{\longrightarrow} (d\varphi^*) \stackrel{g_3}{\longrightarrow} (dy^*),
\label{crmicroconj-modified}
\end{equation}
with obvious notations for linear spaces.

\subsection{Gluing the halves together}
\label{subsec:gluing}

Comparing (\ref{yphi}) with (\ref{f3}), we see that $f_3$ and~$g_3$ are related by matrix transposing:
\begin{equation}
g_3=f_3^{\mathrm T}.
\label{g3f3}
\end{equation}
This remarkable observation is the key for joining together our complexes \eqref{crmicro} and~\eqref{crmicroconj-modified}. Moreover, comparing the formulas \eqref{f4_1} and~\eqref{f4_2} with \eqref{g2_1} and~\eqref{g2_2}, and also \eqref{f5} with \eqref{g1_1} and~\eqref{g1_2}, we find that $f_4$ and $f_5$ are nothing else than $g_2$ and $g_1$ transposed :
\begin{equation}
f_4=g_2^{\mathrm T},\quad f_5=g_1^{\mathrm T}.
\label{f4g2,f5g1}
\end{equation}

We can thus write our complex~\eqref{complex_formal} in a slightly less formal way:
\begin{multline}
0\longrightarrow \mathfrak{psl}(2,\mathbb C) \stackrel{f_1}{\longrightarrow} (dz) \oplus (ds) \stackrel{f_2}{\longrightarrow} (dy)
\stackrel{f_3}{\longrightarrow} (d\varphi)\\
\stackrel{f_4}{\longrightarrow} (d\alpha)\oplus(d\beta) \oplus (dt) \stackrel{f_5}{\longrightarrow} \mathfrak{so}(3,\mathbb C)^* \longrightarrow 0.
\label{complex2}
\end{multline}
Here, $(d\alpha)$, $(d\beta)$, $(dt)$ and~$\mathfrak{so}(3,\mathbb C)^*$ can be considered just as convenient notations for some spaces of column vectors which are in an obvious sense dual to our spaces $(d\alpha^*)$, $(d\beta^*)$ $(dt^*)$ and~$\mathfrak{so}(3,\mathbb C)$ respectively; instead of $\mathfrak{so}(3,\mathbb C)^*$, we could also write $\mathfrak{psl}(2,\mathbb C)^*$, because of the well-known isomorphism between these Lie algebras.

\begin{remark}
\label{rem:right}
We have thus finished the different proof of theorem~\ref{th:complex_formal}: the equalities $f_4\circ f_3=0$ and $f_5\circ f_4=0$ follow by differentiation from theorem~\ref{th-G-global}, using \eqref{g3f3} and the definitions~\eqref{f4g2,f5g1}.
\end{remark}

To finish this section, we think it reasonable to write our complex \eqref{complex_formal} and~\eqref{complex2} in a still more informal and informative way:
\begin{multline}
0 \to \mathfrak{psl}(2,\mathbb C)\stackrel{f_1}{\to}
\left( \begin{smallmatrix}
\textrm{inner vertex}\\ \textrm{coordinate}\\ \textrm{differentials } dz \\ \textrm{and boundary}\\ \textrm{component}\\ \textrm{sways }ds
\end{smallmatrix} \right) \stackrel{f_2}{\to}
\left( \begin{smallmatrix}
\textrm{differentials } dy \\ \textrm{in all tetrahedra}
\end{smallmatrix} \right) 
\\ \stackrel{f_3}{\to} \left( \begin{smallmatrix}
\textrm{differentials } d\varphi \\ \textrm{for all inner edges}\\ \textrm{and some} \\ \textrm{boundary edges}
\end{smallmatrix} \right) \stackrel{f_4}{\to}
\left( \begin{smallmatrix}
\textrm{inner vertex}\\ \textrm{``conjugate coordinate}\\ \textrm{differentials'' } d\alpha \textrm{ and } d\beta \\ \textrm{and boundary component}\\ \textrm{``conjugate sways'' }dt
\end{smallmatrix} \right) \stackrel{f_5}{\to} \mathfrak{so}(3,\mathbb C)^* \to 0 \, .
\label{complex}
\end{multline}

\section{Torsion and a set of invariants}
\label{sec:inv}

The vector spaces in our complex~\eqref{complex_formal} (which we write also in the form \eqref{complex2} or~\eqref{complex}) are spaces of column vectors, which means that they have chosen preferred bases; they are called thus \emph{based} vector spaces. Basis vectors correspond to either triangulation simplexes (vertices, edges, tetrahedra) or some naturally chosen generators of the Lie algebra (formulas \eqref{basis_psl} and~\eqref{basis_so}).

\begin{remark}
\label{rem:basis_order}
As stated in the beginning of section~\ref{sec:complexes}, we are constructing a set of invariants where every individual invariant corresponds to an \emph{ordered} set~$\mathcal D$ of ``marked'' boundary edges. Note though that we \emph{do not} specify the order of basis vectors corresponding to other triangulation simplexes, which will soon result in our invariants being defined up to an overall sign.
\end{remark}

We say that a \emph{$\tau$-chain} is chosen in a complex~$C=(0\to V_0 \stackrel{f_1}{\to} V_1 \stackrel{f_2}{\to} \dots )$
of based vector spaces~$V_i$ if a collection~$\alpha_i$ of basis vectors is chosen in each~$V_i$; the complement of this collection is denoted~$\overline\alpha_i$. To a $\tau$-chain, a collection of submatrices of~$f_i$ corresponds in the following way: the rows for the submatrix of~$f_i$ correspond to~$\alpha_i$, while the columns --- to~$\overline\alpha_{i-1}$. The $\tau$-chain is called \emph{nondegenerate} if all these submatrices are square and nondegenerate.

\begin{lemma}
\label{lemma:acyclic}
A chain complex over a field admits a nondegenerate $\tau$-chain if and only if it is acyclic, i.e., all its homologies are zero.
\qed
\end{lemma}

The proof of this lemma, as well as theorem~\ref{th:torsion} below, can be found e.g.\ in the monograph~\cite{turaev}.

For an acyclic complex~$C$, its (Reidemeister) \emph{torsion} is the following alternated product:
\begin{equation}
\tau(C) \stackrel{\rm def}{=} \prod_i (\minor f_i)^{i+1},
\label{torsion_def}
\end{equation}
where the minors are determinants of the submatrices in a nondegenerate $\tau$-chain. This makes sense due to the following classical theorem:

\begin{theorem}
\label{th:torsion}
Up to a sign, $\tau(C)$ does not depend on the choice of a nondegenerate $\tau$-chain. \qed
\end{theorem}

Thus, the torsion of our complex~\eqref{complex_formal} written for a certain set~$\mathcal D$, defined up to a sign, is
\begin{equation}
\tau_{\mathcal D} = \frac{\minor f_1 \, \minor f_3 \, \minor f_5}{\minor f_2 \, \minor f_4},
\label{torsion}
\end{equation}
if \eqref{complex_formal} has a nondegenerate $\tau$-chain. Actually, a typical situation is that it has such chain for some sets~$\mathcal D$ while does not for other~$\mathcal D$. The aim of the following theorem is to provide the most uniform approach to the complexes for all~$\mathcal D$, and to extend the definition of torsion to the case where a nondegenerate $\tau$-chain does not exist.

\begin{theorem}
\label{th:minors_nonzero}
A $\tau$-chain for complex~\eqref{complex_formal} can be chosen in such way that all minors, except maybe $\minor f_3$, will be nonzero. Moreover, these four minors can be chosen in such way that they do not depend on~$\mathcal D$.
\end{theorem}

\begin{proof}
We will use the notations of formula~\eqref{complex2}. Consider first the case where $\partial M$ is nonempty. 

For $\minor f_1$, we choose the three basis vectors in space~$(ds)$ corresponding to the sways of one --- call it ``first'' --- boundary component, which gives at once $\minor f_1=1$. 

Then, the subspace of $(dz)\oplus (ds)$ corresponding to the sways of other boundary components and all inner coordinate differentials remains for the columns of~$\minor f_2$, and we note that the restriction of~$f_2$ on this subspace is \emph{injective}: as the first boundary component is fixed, and $dy_{ijkl}=0$ in every tetrahedron means that if three of its vertex coordinates are fixed, the fourth one is fixed as well, it follows that the preimage of zero, for the remaining part of~$f_2$, is only zero.

This remaining part of~$f_2$ is a rectangular matrix ($f_2$ minus three its columns), and as it gives an injective linear mapping, we can choose a minimal subset of its rows such that that the submatrix with only these rows left is still injective. It is quite easy to see that such submatrix must be square and nondegenerate, so we choose it as the submatrix corresponding to $\minor f_2$.

Going now to the right end of the complex, we will argue in terms of the conjugate matrices $g_1=f_5^{\mathrm T}$ and $g_2=f_4^{\mathrm T}$. For $\minor g_1$, we choose again the three basis vectors in space~$(dt^*)$ corresponding to the sways of the first boundary component. Then, not only the remaining part of~$g_2$ --- without the three columns --- gives an injective linear mapping, but also we can leave in it only the rows corresponding to \emph{inner} edges: fixing the lengths of just inner edges, together with the immobility of the first boundary component, is obviously enough for the immobility of all inner vertices and all other (rigid!) boundary components. So we can choose here again, like we did for $\minor g_2$, a minimal subset of rows, but this time with the additional requirement that they are inner --- and thus we can choose $\minor g_2$, or equivalently $\minor f_4$ not depending on the chosen set~$\mathcal D$ of boundary edges.

Note that we have chosen the other three minors, not dealing with edges at all, in an obviously independent from~$\mathcal D$ way.

It remains to note that if $\partial M$ is empty, then the previous reasoning is still valid if we choose, for instance, for $\minor f_1$ the three basis vectors in space~$(dz)$ corresponding to the coordinates of three vertices of some two-dimensional face in the triangulation, and for $\minor g_1$ --- the three basis vectors in space~$(d\alpha^*) \oplus (d\beta^*)$ corresponding to, say $d\alpha_i^*$, $d\beta_i^*$ and $d\alpha_i^*$ for some edge~$ij$.
\end{proof}

Due to theorem~\ref{th:minors_nonzero}, we can --- and will --- assume that, for a given triangulated manifold~$M$, the minors of $f_1$, $f_2$, $f_4$ and~$f_5$ are always calculated in one standard way. This fixes also the basis vectors corresponding to the columns of $\minor f_3$, namely those not used for the rows of $\minor f_2$, as well as the basis vectors corresponding to the rows of $\minor f_3$, namely those not used for the columns of $\minor f_4$. The thus obtained $\minor f_3$ is the only one to depend on~$\mathcal D$, and it can turn into zero, which is equivalent (as one can easily see) to complex~\eqref{complex2} being not acyclic. Even in this case, we define the torsion by formula~\eqref{torsion}.

\begin{theorem}
\label{th:invD}
The quantity
\begin{equation}
I_{\mathcal D} = \frac{\tau_{\mathcal D}}{2 \prod\nolimits'\zeta_{ij}^2} ,
\label{invD}
\end{equation}
where the dashed product goes over all inner edges\footnote{Note that our definition~\eqref{invD} slightly differs from~\cite[formula~(50)]{kkm}, where also $\zeta_{ij}^2$ corresponding to boundary edges outside~$\mathcal D$ were included in the product. Our present definition is more convenient for uniting all~$I_{\mathcal D}$ in a ``generating function'', see section~\ref{sec:genfun}.}, is an invariant of manifold~$M$ with the fixed boundary triangulation and given set~$\mathcal D$ of marked boundary edges. 
\end{theorem}

\begin{proof}
As we already mentioned in subsection~\ref{subsec:pachner_moves}, the transition between different triangulations of the interior of~$M$, given a fixed triangulation of~$\partial M$, is achieved by a sequence of \emph{relative} Pachner moves --- moves not changing the boundary triangulation. The proof of this for one specific sort of boundary (specially triangulated torus) has been presented in~\cite[Theorem~1]{dkm}, and it is an easy exercise to make obvious changes so that it will work in the general case.

On the other hand, the proof that \eqref{invD} does not change under relative Pachner moves just repeats the proof of~\cite[Theorem~7]{kkm}.
\end{proof}

\begin{remark}
\label{rem:ordering_D}
The invariant~\eqref{invD} is determined up to a sign depending on the ordering of vertices, edges and tetrahedra used when calculating the minors in~\eqref{torsion}. One can see, however, that if, for a given~$M$ and its boundary triangulation,
\begin{itemize}
	\item a fixed ordering of boundary edges is given, and every set~$\mathcal D$ inherits, by definition, this ordering, and
	\item in the ordering of all edges, boundary edges by definition precede inner edges,
\end{itemize}
then the collection of invariants~\eqref{invD}, for \emph{all}~$\mathcal D$, is determined up to \emph{one overall} sign.
\end{remark}

\begin{remark}
\label{rem:1/2}
We introduced the factor~$1/2$ in~\eqref{invD}\footnote{which was not done in paper~\cite{kkm}} so as to make the invariant of sphere~$S^3$ (closed manifold, so $\mathcal D = \varnothing$) equal to~$1$. This invariant can be calculated directly from formula~~\eqref{invD} using, e.g., the simplest triangulation of two tetrahedra.
\end{remark}

\section{Generating functions of Grassmann variables}
\label{sec:genfun}

\subsection{Generating functions for a rectangular matrix}
\label{subsec:genfun_matrices}

Here we develop a version\footnote{This is a simplified construction as compared to paper~\cite{tqft2} where we were dealing with \emph{sums} of matrices (extended if necessary by additional rows and/or columns of zeros), while in the present paper, we are dealing just with their \emph{concatenations}.} of our construction of a generating function of anticommuting variables put in correspondence to a matrix~$A$. In this subsection, $A$~is an arbitrary matrix whose entries are complex-valued expressions, with the only condition that the number of rows is not smaller than the number of columns.

With each row~$k$ of~$A$, we associate a Grassmann generator~$a_k$, while with the whole matrix~$A$ --- the \emph{generating function} defined as
\begin{equation}
\mathbf f_A = \sum_{\mathcal C} \det A|_{\mathcal C} \prod_{k\in \mathcal C} a_k,
\label{genfun}
\end{equation}
where $\mathcal C$ runs over all subsets of the set of rows of the cardinality equal to the number of columns; $A|_{\mathcal C}$ is the square submatrix of~$A$ containing all rows in~$\mathcal C$; the order of~$a_k$ in the product is the same as the order of rows in~$A|_{\mathcal C}$ (e.g., the most natural --- increasing --- order of~$k$'s in both).

\begin{lemma}
\label{lemma:multiplication}
Let $C$ be the concatenation of matrices $A$ and~$B$ having the equal number of rows: $C=\begin{pmatrix} A & B \end{pmatrix}$. Then
$$
\mathbf f_C = \mathbf f_A \mathbf f_B .
$$
\end{lemma}

\begin{proof}
The lemma easily follows from the expansion of the form
\begin{equation}
\minor C = \sum \pm \minor A \, \minor B,
\label{pm}
\end{equation}
known from linear algebra, for every minor of~$C$ having the full number of columns.
\end{proof}

Let there be now a subset~$\mathcal I$ of ``marked'' rows of~$A$. We call the rows in~$\mathcal I$ \emph{inner}, while the rest of rows --- \emph{outer}, and we define the \emph{generating function of matrix~$A$ with the set~$\mathcal I$ of inner edges} as
\begin{equation}
{}_{\mathcal I}\mathbf f_A = \sum_{\mathcal C \supset \mathcal I} \det\nolimits' A|_{\mathcal C} \prod_{k\in \mathcal C \setminus \mathcal I} a_k.
\label{genfunI}
\end{equation}
Here $\det\nolimits'$ means that, unlike in~\eqref{genfun}, we are changing the order of $A$'s rows in the following way: all inner rows are brought to the bottom of the matrix; the order of rows within the set~$\mathcal I$ and its complement is conserved; the order of~$a_k$'s in the product (where $k$ belongs to the mentioned complement) is the same as the order of rows~$k$.

\begin{lemma}
\label{lemma:integration}
The generating function of matrix~$A$ with the set~$\mathcal I$ of inner edges is the following Berezin integral of the usual generating function:
\begin{equation}
{}_{\mathcal I}\mathbf f_A = \int \mathbf f_A \prod_{l\in\mathcal I}^{\leftarrow} da_l,
\label{f_inner}
\end{equation}
the arrow above the product means that the differentials are written in the reverse (with respect to the order of rows in~$A$) order.
\end{lemma}

\begin{proof}
First, we note that only those terms in~$\mathbf f_A$ survive the integration in the r.h.s.\ of~\eqref{f_inner} which contain all the~$a_k$ for~$k\in \mathcal I$. We take the function~$\mathbf f_A$ as defined in~\eqref{genfun}, leave only the mentioned terms in it, and note that none of them is changed if we bring both the rows~$k$ in~$A$ for all~$k\in \mathcal I$ to the bottom of the matrix and the corresponding generators~$a_k$ to the right in the product\footnote{because any elementary permutation of rows brings a minus sign which cancels out with the minus brought by the corresponding permutation of~$a_k$'s}, neither changing the order within~$\mathcal I$ nor within its complement. Then, the integration in~\eqref{f_inner} just takes away the~$a_k$ for~$k\in \mathcal I$, as required.
\end{proof}

\subsection{Generating function for invariants of a manifold with triangulated boundary}
\label{subsec:genfun_invariants}

To produce a generating function whose coefficients are the invariants~\eqref{invD}, we take the following matrix:
\begin{equation}
A=\frac{1}{2\prod\nolimits' \zeta_{ij}^2} \, \frac{\minor f_1\, \minor f_5}{\minor f_2\, \minor f_4} \, \tilde f_3 \, ,
\label{A}
\end{equation}
where $\tilde f_3$ is the submatrix of the Jacobian matrix $(\partial \varphi_{ij} / \partial y_a)$ containing the columns and rows corresponding to tetrahedra~$a$ and edges~$ij$ not used in $\minor f_2$ and $\minor f_4$ respectively. In particular, $\tilde f_3$ contains the rows corresponding to \emph{all boundary edges}.

Looking at the dimensions in formula~\eqref{complex_formal}, one can deduce that $\tilde f_3$ has $( N_1-2N'_0-3m+3 )$ rows and $(N_3-N'_0-3m+3)$ columns. Hence, the fact that $A$ has not less rows than columns follows from lemma~\ref{lemma:boundary_edges}.

As the rows of~\eqref{A} correspond to triangulation edges, so do the Grassmann variables on which $\mathbf f_A$ depends.  

We want a function depending only on \emph{boundary} Grassmann variables, so we pass on to function ${}_{\mathcal I}\mathbf f_A$ where $\mathcal I$ is the set of those inner edges that correspond to the rows of~$\tilde f_3$; we call it the \emph{generating function for invariants of manifold~$M$ with triangulated boundary} and denote as
\begin{equation}
\mathbf I_M \stackrel{\rm def}{=} {}_{\mathcal I}\mathbf f_A = \int \mathbf f_A \prod^{\leftarrow}_{\textrm{edges in }\mathcal I} da_{ij}\, .
\label{I_M}
\end{equation}

According to remark~\ref{rem:ordering_D}, our generating functions are determined up to a sign.

\begin{remark}
\label{rem:tetrahedron}
One can see now that the expression~\eqref{fbf} is nothing but $2\,\mathbf I_M$ for $M$ being a single tetrahedron considered as a manifold with boundary. Moreover, it will become clear soon (remark~\ref{rem:pentagon_proved}) that the l.h.s.\ and r.h.s.\ of~\eqref{fffff} are the~$2\,\mathbf I_M$ for $M$ being the l.h.s.\ and r.h.s.\ respectively of Pachner move~$2\to 3$.

In this paper, we reserve the name ``tetrahedron function'' for the expression~\eqref{fbf} --- the \emph{doubled} generating function of invariants for a single tetrahedron.
\end{remark}

\section{Changing the boundary triangulation, and a lemma about the state sum}
\label{sec:boundary_change}

If we change the boundary triangulation of manifold~$M$, the new function~$\mathbf I_M$ can be expressed in terms of the old one. Any boundary triangulation change can be achieved using a sequence of \emph{two-dimensional} Pacher moves. Namely, there are moves $1\to 3$, $2\to 2$ and $3\to 1$, which correspond to gluing a new tetrahedron to the boundary by one, two or three of its faces respectively\footnote{and the faces on the boundary to which the tetrahedron is glued must form a star of a 2-, 1- or 0-simplex respectively}.

\begin{lemma}
\label{lemma:13}
A move $1\to 3$ corresponds to multiplying $\mathbf I_M$ by the tetrahedron function~\eqref{fbf}.
\end{lemma}

\begin{proof}
Neither new inner vertices nor new inner edges appear in this case. So, first, only $f_3^{\rm full}$ changes in formula~\eqref{A}. Second, the change of~$f_3^{\rm full}$ can be described as adding to it the $(6\times 1)$-matrix $A_a = (\partial \varphi_{ij} / \partial y_a)$ written for the new tetrahedron~$a$, with both matrices first extended by zeros in rows and columns corresponding to ``missing'' edges and tetrahedra (the new $f_3^{\rm full}$ will have, of course, three new rows and one column with respect to the old one). Calculating explicitly matrix~$A_a$ and using lemma~\ref{lemma:multiplication}, we see that $\mathbf f_A$ is multiplied by the tetrahedron function. As $\mathbf I_M$, both before and after the move, is the integral~\eqref{I_M} of corresponding $\mathbf f_A$, and the tetrahedron function plays the role of constant with respect to the integration, one comes to the statement of the lemma.
\end{proof}

\begin{lemma}
\label{lemma:22}
A move $2\to 2$ corresponds to multiplying $\mathbf I_M$ by the tetrahedron function~\eqref{fbf} and then integration in the Grassmann variable living on the edge which becomes inner.
\end{lemma}

\begin{proof}
Again, as in the proof of lemma~\ref{lemma:13}, only $f_3^{\rm full}$ changes in formula~\eqref{A}, and this change can be described as adding to it the $(6\times 1)$-matrix $A_a$ (although, this time, the new $f_3^{\rm full}$ will have just \emph{one} new row and one column with respect to the old one). As one boundary edge becomes inner under the move, the multiplication made according to lemma~\ref{lemma:multiplication} must be followed by integration according to lemma~\ref{lemma:integration}.
\end{proof}

\begin{remark}
\label{rem:pentagon_proved}
With lemmas \ref{lemma:13} and~\ref{lemma:22} proved, one can construct the generating functions for the clusters of tetrahedra in l.h.s.\ and r.h.s.\ of Pachner move~$2\to 3$, starting from one tetrahedron and adding more of them. Equation~\eqref{fffff} follows now from theorem~\ref{th:invD}. Note, however, that we have proved~\eqref{fffff} in this way only up to a sign.
\end{remark}

The remaining Pachner move on boundary is $3\to 1$.

\begin{lemma}
\label{lemma:31}
Let a Pachner move~$3\to 1$ on boundary be done by gluing a tetrahedron~$jkli$ to the boundary in such way that vertex~$i$ becomes inner. Then the new~$\mathbf I_{M}$ is obtained from the old one by any of the following ways:
\begin{equation}
\mathbf I_{M}^{\rm new} = \frac{1}{\zeta_{ij}\zeta_{kl}} \int \mathbf I_{M}^{\rm old} \, da_{ij} = \frac{1}{\zeta_{ik}\zeta_{lj}} \int \mathbf I_{M}^{\rm old} \, da_{ik} = \frac{1}{\zeta_{il}\zeta_{jk}} \int \mathbf I_{M}^{\rm old} \, da_{il} \, .
\label{31}
\end{equation}
\end{lemma}

\begin{proof}
A move~$3\to 1$ is the (two-sided) inverse of~$1\to 3$, and in our case $1\to 3$ means gluing a tetrahedron~$ijkl$ (oppositely oriented to~$jkli$) by its face~$jkl$. So, it follows from lemma~\ref{lemma:13} that the coefficient at~$a_{ij}$ in~$\mathbf I_{M}$ before the move~$3\to 1$ must be $\zeta_{ij}\zeta_{kl}$ times the whole~$\mathbf I_{M}$ after the move~$3\to 1$, and the integration in~$da_{ij}$ in~\eqref{31} singles out exactly this coefficient. Other equalities in~\eqref{31} appear if we take edge $ik$ or~$il$ instead of~$ij$.
\end{proof}

To finish this section, we use the technique developed here in proving the following lemma.

\begin{lemma}
\label{lemma:t_i}
The state sum~\eqref{s} of a triangulated closed oriented connected manifold~$M$ is the doubled generating function~$\mathbf I_M$ if the triangulation has no inner vertices and $\partial M$~has exactly one connected component; otherwise, it vanishes.
\end{lemma}

\begin{proof}
It is an easy exercise to show, using the same kind of reasoning as in lemmas \ref{lemma:13} and~\ref{lemma:22}, that
\begin{equation*}
(\text{the generating function for matrix }f_3^{\rm full}) = \int \prod \mathbf f_{klmn} \prod\nolimits' da_{ij},
\end{equation*}
where $f_3^{\rm full}=(\partial \varphi_{ij} / \partial y_a) $ is the Jacobian matrix involving \emph{all} tetrahedra~$a$ and \emph{all} edges~$ij$; the first product goes over all tetrahedra~$klmn$, and the dashed product --- over all inner edges~$ij$.

If now the triangulation has no inner vertices and $\partial M$~has exactly one connected component, the minors of $f_1$, $f_2$, $f_4$ and~$f_5$, chosen as in the proof of theorem~\ref{th:minors_nonzero}, are all equal to unity; in the case of $f_2$ and~$f_4$ --- because they are of zero size. This also implies $\tilde f_3=f_3^{\rm full}$ for the function~$\tilde f_3$ defined in subsection~\ref{subsec:genfun_invariants}. Substituting this all into~\eqref{A} and using the definition~\eqref{I_M} of~$\mathbf I_M$ proves the lemma for this case.

If the triangulation does have inner vertices or there are more than one boundary components, a nontrivial $\minor f_2$ appears, which implies that the rank of~$f_3^{\rm full}$ is less than~$N_3$ --- the number of all tetrahedra, and the generating function for matrix~$f_3^{\rm full}$ is the identical zero.
\end{proof}

\section{Gluing manifolds over a boundary component}
\label{sec:gluing}

Our theory deserves the name TQFT if it provides a means to express the generating function of invariants of the result of gluing two manifolds in terms of the generating function of two those manifolds. In this section, we consider this problem for manifolds $M_1$ and~$M_2$ glued over one component of their boundaries; the result of gluing is denoted~$M$; the mentioned boundary component --- closed connected triangulated surface --- is denoted~$\Gamma$; if it is desirable to emphasize that it belongs, specifically, to $M_1$ or~$M_2$, we also denote it (or its copies) as $\Gamma_1 \subset M_1$ and~$\Gamma_2 \subset M_2$.

\subsection{Maximal tree of triangles in $\Gamma$, virtual tetrahedra and virtual edges}
\label{subsec:tree}

We will adopt the following condition on the triangulation of~$\Gamma$: there exists such ordering $i_1,\dots,i_n$ of all vertices in~$\Gamma$ that:
\begin{itemize}
	\item $i_1i_2i_3$ is one of the triangles in the triangulation of~$\Gamma$, we call this triangle~$\Delta_1$;
	\item there exist also such triangles $\Delta_2,\dots,\Delta_{n-2}$ in the triangulation of~$\Gamma$ that, for every $m=4,\dots,n$, triangle~$\Delta_{m-2}$ has $i_m$ as one of its vertices, and also $\Delta_{m-2}$ has a common edge with one of the ``previous'' triangles $\Delta_1,\dots, \allowbreak \Delta_{m-3}$.
\end{itemize}
This technical condition is just for making our work in this section easier; there exist of course plenty of triangulations of any closed orientable two-dimensional~$\Gamma$ satisfying this condition, and we will use some of them in section~\ref{sec:examples}.

We define a \emph{maximal tree of triangles} in~$\Gamma$ as the collection of such triangles $\Delta_1,\dots,\Delta_{n-2}$. We are also going to construct a sequence of \emph{virtual tetrahedra} $t_1,\dots,t_{n-3}$ in the following way. By definition, $t_1$ has $\Delta_1$ and~$\Delta_2$ as two of its faces; two other faces are new --- not present in~$\Gamma$; as $t_1$ has six edges, while $\Delta_1$ and~$\Delta_2$ together --- only five, one edge in~$t_1$ is also new.

Then we proceed by induction: for any $m=2,\dots,n-3$, two of the faces of tetrahedron~$t_m$ are, by definition, $\Delta_{m+1}$ and that triangle~$\Delta'_m$ in the common boundary of the already constructed tetrahedra but not belonging to~$\Gamma$:
$$
\Delta'_m \subset \partial (t_1 \cup \dots \cup t_{m-1}) \setminus \Gamma ,
$$
which has a common edge with~$\Delta_{m+1}$; two other faces, and one edge, are new. Exactly one such triangle~$\Delta'_m$ exists, of course, at any step~$m$; note also that exactly half of (the two-dimensional faces in) the boundary of $t_1 \cup \dots \cup t_{m-1}$ belongs to~$\Gamma$.

After the last step $m=n-3$, we obtain a cluster of tetrahedra having $\Delta_1 \cup \dots \cup \Delta_{n-2}$ as half of its boundary.

According to our construction, while adding every new virtual tetrahedron, we added also a new edge. We have thus obtained a collection of $n-2$ such edges, and we call them \emph{virtual edges}.

Our idea is to express the algebraic complex~\eqref{complex} for~$M$ in terms of algebraic complexes for $M_1$, $M_2$ and~$\Gamma$. We expect all these complexes to be of the same nature as~\eqref{complex}; but $\Gamma$ is just a surface, containing no tetrahedra. So what we do is inflating $\Gamma$ with two (oppositely oriented copies of) clusters of ``virtual tetrahedra'' described above: we take two copies $\Gamma_1$ and~$\Gamma_2$ of~$\Gamma$, glue one cluster to~$\Gamma_1$ and the other to~$\Gamma_2$, then glue the other halves of boundaries of these clusters together, and also identify the triangles in $\Gamma_1$ and~$\Gamma_2$ not belonging to our maximal tree. We call the result ``inflated~$\Gamma$'' and denote as~$\hat{\Gamma}$.

Note that we have thus identified the two copies of each virtual edge, so their number remains $n-2$.

The manifold obtained by gluing $M_1$ and~$M_2$ to the two sides of~$\hat{\Gamma}$ is of course again the same~$M$, but with two additional clusters of tetrahedra in its triangulation. We call this triangulated manifold ``inflated~$M$'' and denote as~$\hat M$.

\subsection{Enlarged complex: description}
\label{subsec:ec_description}

We consider the following algebraic complex, which is the complex~\eqref{complex} for~$\hat M$ with additional direct summands in some terms:
\begin{multline}
0 \to 3 \times \mathfrak{psl}(2,\mathbb C)\stackrel{f_1}{\to}
\left( \begin{smallmatrix}
dz \textrm{ for inner}\\ \textrm{vertices of } \hat M,\\ \textrm{boundary}\\ \textrm{component}\\ \textrm{sways } ds \textrm{ for } \hat M \\ \textrm{and}\\ 2\times (\textrm{sways } ds \textrm{ of } \Gamma)
\end{smallmatrix} \right) \stackrel{f_2}{\to}
\left( \begin{smallmatrix}
\textrm{differentials}\\ dy \\ \textrm{in all tetrahedra}
\end{smallmatrix} \right) \\
\stackrel{f_3}{\to} \left( \begin{smallmatrix}
\textrm{differentials}\\ d\varphi \\ \textrm{for all inner}\\ \textrm{and some boundary}\\ \textrm{edges of } \hat M
\end{smallmatrix} \right) \stackrel{f_4}{\to}
\left( \begin{smallmatrix}
d\alpha \textrm{ and } d\beta \textrm{ for inner}\\ \textrm{vertices of } \hat M,\\ \textrm{boundary component}\\ \textrm{conjugate sways } dt \textrm{ for } \hat M \\ \textrm{and}\\ 2\times (\textrm{conjugate sways } dt \textrm{ of } \Gamma)
\end{smallmatrix} \right) \stackrel{f_5}{\to} 3 \times \mathfrak{so}(3,\mathbb C)^* \to 0 \, .
\label{complex_enlarged}
\end{multline}

Because of the additional direct summands in~\eqref{complex_enlarged}, we must give new definitions for the mappings $f_1,\dots,f_5$.

To begin, it is convenient and relevant to assign subscripts to the three copies of~$\mathfrak{psl}(2,\mathbb C)$ (coming after the left zero), denoting them as $\mathfrak{psl}(2,\mathbb C)_{\hat M}$, $\mathfrak{psl}(2,\mathbb C)_{M_1}$ and~$\mathfrak{psl}(2,\mathbb C)_{M_2}$. Similarly, we denote $ds_{\Gamma_1}$ and~$ds_{\Gamma_2}$ two copies of sways of surface~$\Gamma$ in the second nonzero term from the left in~\eqref{complex_enlarged}. We also denote $dt_{\Gamma_1}$ and~$dt_{\Gamma_2}$ two copies of surface~$\Gamma$ conjugate sways in the second nonzero term from the right, and $\mathfrak{so}(3,\mathbb C)_{\hat M}^*$, $\mathfrak{so}(3,\mathbb C)_{M_1}^*$ and~$\mathfrak{so}(3,\mathbb C)_{M_2}^*$ --- the three copies of~$\mathfrak{so}(3,\mathbb C)^*$ in the term before the right zero.

By definition, $f_1$ acts as follows:
\begin{itemize}
\item $\mathfrak{psl}(2,\mathbb C)_{\hat M}$ acts naturally --- according to~\eqref{f1} --- on~$dz_i$ for all inner vertices~$i$ of~$\hat M$ (including vertices in~$\Gamma$), and on the sways~$ds_{\kappa}$ of boundary components~$\kappa$ of~$\hat M$ (where $\Gamma$ does not enter);
\item $\mathfrak{psl}(2,\mathbb C)_{M_1}$ acts naturally on $dz_i$ for all inner vertices~$i$ of~$M_1$ (but not $M_2$ and not~$\Gamma$), sways of boundary components of~$M_1$ \emph{without}~$\Gamma$, and the first copy~$ds_{\Gamma_1}$ of sways of~$\Gamma$;
\item $\mathfrak{psl}(2,\mathbb C)_{M_2}$ acts naturally on $dz_i$ for all inner vertices~$i$ of~$M_2$, sways of boundary components of~$M_2$ without~$\Gamma$, and the second copy~$ds_{\Gamma_2}$ of sways of~$\Gamma$.
\end{itemize}

Mapping~$f_2$ acts, by definition, as follows:
\begin{itemize}
\item $dz_i$, for all inner vertices~$i$ in~$\hat M$ (including those in~$\Gamma$), act naturally on~$dy_a$ in the adjoining tetrahedra~$a$, that is, according to formula~\eqref{f2_1};
\item the same applies to the sways~$ds_{\kappa}$ of boundary components~$\kappa$ of~$\hat M$, which act according to \eqref{f2_1} and~\eqref{f2_2};
\item sways~$ds_{\Gamma_1}$ act \emph{only} on~$dy_a$ in tetrahedra~$a$ belonging to~$M_1$ (but neither tetrahedra in~$M_2$ nor virtual tetrahedra);
\item sways~$ds_{\Gamma_2}$ act only on~$dy_a$ in tetrahedra~$a$ belonging to~$M_2$.
\end{itemize}

Mapping~$f_3$ just acts in the same way as in~\eqref{complex}, i.e., according to~\eqref{f3}.

Mapping~$f_4$ acts, by definition, as follows:
\begin{itemize}
	\item all differentials~$d\varphi_{ij}$ in the space before arrow~$f_4$ act according to~\eqref{f4_1} on $d\alpha_i$ and~$d\beta_i$ for all inner vertices~$i$ of~$\hat M$ and according to \eqref{f4_1} and~\eqref{f4_2} --- on conjugate sways~$dt_{\kappa}$ of boundary components~$\kappa$ of~$\hat M$;
	\item the differentials~$d\varphi_{ij}$ for edges~$ij$ belonging to~$M_1$ act also, according to \eqref{f4_1} and~\eqref{f4_2}, on conjugate sways~$dt_{\kappa}$ of boundary components~$\kappa$ of~$M_1$;
	\item similarly, the differentials~$d\varphi_{ij}$ for edges~$ij$ belonging to~$M_2$ act also on conjugate sways~$dt_{\kappa}$ of boundary components~$\kappa$ of~$M_2$.
\end{itemize}

Finally, mapping~$f_5$ acts in the following way, symmetric to~$f_1$:
\begin{itemize}
\item contributions to $\mathfrak{so}(3,\mathbb C)_{\hat M}^*$, namely in the sums according to~\eqref{f5}, are made by $d\alpha_i$ and~$d\beta_i$ for all inner vertices~$i$ of~$\hat M$, and conjugate sways~$dt_{\kappa}$ of boundary components~$\kappa$ of~$\hat M$;
\item contributions to $\mathfrak{so}(3,\mathbb C)_{M_1}^*$ are made by $d\alpha_i$ and~$d\beta_i$ for all inner vertices~$i$ of~$M_1$ (but not $M_2$ and not~$\Gamma$), conjugate sways of boundary components of~$M_1$ without~$\Gamma$, and the first copy~$dt_{\Gamma_1}$ of conjugate sways of~$\Gamma$;
\item contributions to $\mathfrak{so}(3,\mathbb C)_{M_2}^*$ are made by $d\alpha_i$ and~$d\beta_i$ for all inner vertices~$i$ of~$M_2$, conjugate sways of boundary components of~$M_2$ without~$\Gamma$, and the second copy~$dt_{\Gamma_2}$ of conjugate sways of~$\Gamma$.
\end{itemize}

\subsection{Enlarged complex in terms of $M$}
\label{subsec:ec_M}

We want to compare the torsion of complex~\eqref{complex_enlarged} with the torsion of the usual complex~\eqref{complex} written for~$\hat M$. To do so, we calculate the torsion of~\eqref{complex_enlarged} choosing $\minor f_1$ in the following special way: we take the $\minor f_1$ which we would use for complex~\eqref{complex}, assume that $\mathfrak{psl}(2,\mathbb C)$ in~\eqref{complex} will correspond to~$\mathfrak{psl}(2,\mathbb C)_{\hat M}$ in~\eqref{complex_enlarged}, and extend this $\minor f_1$ by the rows corresponding to $ds_{\Gamma_1}$ and~$ds_{\Gamma_2}$ and, of course, by the columns corresponding to $\mathfrak{psl}(2,\mathbb C)_{M_1}$ and~$\mathfrak{psl}(2,\mathbb C)_{M_2}$. The appearing ``large'' $\minor f_1$, if written in the most natural way, has a triangular block structure with two of three diagonal blocks being $3\times 3$ identity matrices; it is thus evident that it is simly equal to the original ``small'' $\minor f_1$.

We also choose the ``large'' $\minor f_5$ in a perfectly symmetric way (here, of course, rows are interchanged with columns) and come to the conclusion that it is also equal to the ``small'' $\minor f_5$.

\begin{lemma}
\label{lemma:ec_M}
The torsion of complex~\eqref{complex_enlarged} is equal to the torsion of complex~\eqref{complex} written for~$\hat M$.
\end{lemma}

\begin{proof}
It remains to choose the very same minors of $f_2$, $f_3$ and~$f_4$ for~\eqref{complex_enlarged} as have been chosen for~\eqref{complex}.
\end{proof}

\subsection{Enlarged complex in terms of $M_1$ and $M_2$}
\label{subsec:ec_M1M2}

Here we start from given minors (used in formula~\eqref{torsion} for torsion) chosen for complexes~\eqref{complex} written for $M_1$ and~$M_2$. Recall that, according to theorem~\ref{th:minors_nonzero}, all minors except $\minor f_3$ can be chosen once and for all, not depending on the choice of marked boundary edges.

From now on, we supply minors belonging to $M_1$ and~$M_2$ with superscripts, writing them as $\minor f_i^{(1)}$ or~$\minor f_i^{(2)}$ respectively, $i=1,\dots,5$. We are going to build minors for complex~\eqref{complex_enlarged} --- for which we reserve the notation $\minor f_i$ --- extending the direct sums of these minors\footnote{To be exact, the direct sum of corresponding \emph{submatrices} is, of course, taken. It is considered as a submatrix of the corresponding~$f_i$ belonging to~$\hat M$.} belonging to $M_1$ and~$M_2$ by new rows and columns. These ``enlarged'' minors may coincide or not with those in subsection~\ref{subsec:ec_M}.

So, we include in $\minor f_1$ the rows\footnote{in addition, of course, to the rows in $\minor f_1^{(1)}$ and~$\minor f_2^{(2)}$} corresponding to $dz_{i_1}$, $dz_{i_2}$ and~$dz_{i_3}$, where vertices $i_1$, $i_2$ and~$i_3$ have been defined in subsection~\ref{subsec:tree}. This gives
\begin{equation}
\minor f_1 = \minor f_1^{(1)} \, \minor f_1^{(2)} \, \frac{dz_{i_1} \wedge dz_{i_2} \wedge dz_{i_3}}{da \wedge db \wedge dc},
\label{m1}
\end{equation}
where $da$, $db$ and~$dc$ belong to $\mathfrak{psl}(2,\mathbb C)_{\hat M}$ and correspond to the three columns which must also be included in~$\minor f_1$.

Then, we include in $\minor f_2$ the rows corresponding to~$dy_a$ in all tetrahedra~$a$ belonging to one of the clusters by which we inflate~$\Gamma$ as described in subsection~\ref{subsec:tree}. We must also include there the columns corresponding to the rest of vertices in~$\Gamma$, so this gives:
\begin{equation}
\minor f_2 = \minor f_2^{(1)} \, \minor f_2^{(2)} \, \frac{\bigwedge_{\textrm{cluster }1} dy_a}{dz_{i_4} \wedge \dots \wedge dz_{i_n}},
\label{m2}
\end{equation}
where $\bigwedge_{\textrm{cluster }1}$ means, of course, the exterior product over one cluster --- we will call this cluster ``first''.

Now we switch to the other end of complex~\eqref{complex_enlarged} and consider~$\minor f_5$. We include in it the \emph{columns} corresponding to (say) $d\alpha_{i_1}$, $d\beta_{i_1}$ and~$d\alpha_{i_2}$. This gives
\begin{equation}
\minor f_5 = \minor f_5^{(1)} \, \minor f_5^{(2)} \, \frac{da^* \wedge db^* \wedge dc^*}{d\alpha_{i_1} \wedge d\beta_{i_1} \wedge d\alpha_{i_2}},
\label{m5}
\end{equation}
where $da^*$, $db^*$ and~$dc^*$ belong to~$\mathfrak{so}(3,\mathbb C)^*_{\hat M}$.

Then, we include in $\minor f_4$ the columns corresponding to~$d\varphi_{ij}$ for all edges~$ij$ in the maximal tree in~$\Gamma_1$. This gives
\begin{equation}
\minor f_4 = \minor f_4^{(1)} \, \minor f_4^{(2)} \, \frac{d\beta_{i_2} \wedge d\alpha_{i_3} \wedge d\beta_{i_3} \wedge \dots \wedge d\alpha_{i_n} \wedge d\beta_{i_n}}{ \bigwedge_{\textrm{tree }1} d\varphi_{ij}},
\label{m4}
\end{equation}
where $\bigwedge_{\textrm{tree }1}$ corresponds to the mentioned maximal tree in~$\Gamma_1$.

We look now at what remains for $\minor f_3$. Its columns, besides those in $\minor f_3^{(1)}$ and~$\minor f_3^{(2)}$, must correspond to~$dy_a$ for the tetrahedra~$a$ in the second cluster of inflated~$\Gamma$. The number of these tetrahedra is the same as the number of virtual edges and, moreover, these tetrahedra are the only remaining tetrahedra\footnote{as the first cluster of virtual tetrahedra is already involved in $\minor f_2$} containing the virtual edges. This leads to a triangular structure of the remaining part of $\minor f_3$ and to the formula
\begin{equation}
\minor f_3 = \minor f_3^{(1,2)} \, \frac{ \bigwedge_{\textrm{virtual}} d\varphi_{ij}}{ \bigwedge_{\textrm{cluster }2} dy_a},
\label{m3}
\end{equation}
where $\minor f_3^{(1,2)}$, in contrast with formulas \eqref{m1}--\eqref{m4}, is not just a product of two minors belonging to $M_1$ and~$M_2$ separately. It is rather the determinant of the submatrix of~$f_3$ whose columns correspond to all tetrahedra in $M_1$ and~$M_2$ except those involved in $\minor f_2^{(1)}$ and~$\minor f_2^{(2)}$, and whose rows correspond to some inner edges of $M_1$ and~$M_2$ (those not involved in $\minor f_4^{(1)}$ and~$\minor f_4^{(2)}$) and all boundary edges of $M_1$ and~$M_2$ except those in the maximal tree of triangles in~$\Gamma_1$ (as they work already in the rightmost factor in~\eqref{m4}). We denote this submatrix~$B$. It is thus the concatenation of its two parts: $B=\begin{pmatrix} B_1 & B_2 \end{pmatrix}$, belonging to $M_1$ and~$M_2$ respectively.

\subsection{The final formula for generating functions}
\label{subsec:gluing_gf}

We introduce now some more notations. The set of edges in~$\Gamma$ not belonging to the maximal tree of triangles is denoted~$\mathcal F$. As we, according to subsection~\ref{subsec:tree}, identify the triangles in $\Gamma_1$ and~$\Gamma_2$ not belonging to the maximal trees, $\mathcal F$~is not duplicated when gluing together $\Gamma_1$, $\Gamma_2$ and the virtual tetrahedra between them. And the set of inner edges in the maximal tree of triangles, considered as two-dimensional triangulated manifold with boundary, is denoted~$\mathcal G$. To be exact, there are two copies of this set, lying one in~$\Gamma_1$ and the other in~$\Gamma_2$, so we denote them $\mathcal G_1$ and~$\mathcal G_2$ respectively.

\begin{lemma}
\label{lemma:ap}
The alternated product of the rightmost factors in the five formulas \eqref{m1}--\eqref{m3} (with the factors corresponding to minors with odd subscripts taken in the power~$+1$, and with even subscripts --- in the power~$-1$) is equal to
\begin{equation}
\frac{dz_{i_1} \wedge dz_{i_2} \wedge dz_{i_3}}{da \wedge db \wedge dc} \cdots \frac{da^* \wedge db^* \wedge dc^*}{d\alpha_{i_1} \wedge d\beta_{i_1} \wedge d\alpha_{i_2}} = \pm 2 \prod_{{\textrm{tree }1}} \zeta_{ij}^2 ,
\label{ap}
\end{equation}
where the product is taken over all edges in one --- first, for instance --- maximal tree of triangles.
\end{lemma}

\begin{proof}
It is convenient to represent the product in the l.h.s.\ of~\eqref{ap} as the torsion of the following acyclic complex corresponding to the part of~$\hat{\Gamma}$ consisting of the two clusters of tetrahedra, \emph{with each edge in~$\mathcal G_1$ identified with the corresponding edge in~$\mathcal G_2$}. We denote by~$S$ the manifold obtained by gluing together the two copies of the maximal tree of triangles, which is of course homeomorphic to~$S^3$, and by $f_1^S,\dots, f_5^S$ --- the mappings $f_1,\dots, f_5$ acting in the standard way in the complex written for~$S$:
\begin{equation}
0 \to \mathfrak{psl}(2,\mathbb C) \stackrel{f_1^S}{\to} (dz) \stackrel{f_2^S}{\to} (dy) \stackrel{f_3^S}{\to} (d\varphi) \stackrel{f_4^S}{\to} (d\alpha) \oplus (d\beta) \stackrel{f_5^S}{\to} \mathfrak{so}(3,\mathbb C)^* \to 0 \, .
\label{tG}
\end{equation}
It is quite easy to see that the l.h.s.\ of~\eqref{ap} is nothing but the torsion of~\eqref{tG}, after which \eqref{ap} follows from formula~\eqref{invD} and remark~\ref{rem:1/2}.
\end{proof}

\begin{theorem}
\label{th:gluing_gf}
The generating function of invariants for manifold~$M$ --- the result of gluing $M_1$ and~$M_2$ over boundary component~$\Gamma$ --- can be expressed as follows:
\begin{equation}
\mathbf I_M =  \frac{4}{\prod_{\mathcal F \cup \mathcal G} \zeta_{ij}^2} \int \mathbf I_{M_1} \mathbf I_{M_2} \prod_{\mathcal F \cup \mathcal G_2} da_{ij} \, .
\label{gluing_gf}
\end{equation}
It is assumed in the Berezin integral in~\eqref{gluing_gf} that the anticommuting variables living on $\mathcal G_1$ and~$\mathcal G_2$ are \emph{different}, while the rest of anticommuting variables are \emph{identified}.
\end{theorem}

\begin{proof}
The coefficients of~$\mathbf I_M=\mathbf I_{\hat M}$ at various monomials corresponding to various choices of set~$\mathcal D$ of marked edges in~$\partial M$ are invariants calculated according to~\eqref{invD}, with the torsion~$\tau_{\mathcal D}$ calculated according to~\eqref{torsion}. So, the proof of the theorem consists in gathering together:
\begin{itemize}
	\item the factors for minors according to \eqref{m1}--\eqref{m3},
	\item the factors of the type~$\zeta_{ij}^{\pm 2}$ according to which inner edges in~$\hat M$ are new with respect to those in $M_1$ and~$M_2$, and to formula~\eqref{ap},
	\item and the degrees of number~$2$ appearing in the definition~\eqref{invD} of the invariant and in~\eqref{ap}.
\end{itemize}
Except for $\minor f_3$, we obtain thus just numerical factors not depending on~$\mathcal D$. The only special situation appears for $\minor f_3$: as explained after formula~\eqref{m3}, it is the determinant of the concatenation of two matrices, so an expansion of the form~\eqref{pm} holds for it. As also some new edges are declared inner, the result, in terms of generating functions, is obtained according to lemmas \ref{lemma:multiplication} and~\ref{lemma:integration}, which leads exactly to~\eqref{gluing_gf}.
\end{proof}

\begin{remark}
\label{rem:asymmetry}
The asymmetry of formula~\eqref{gluing_gf} with respect to $M_1$ and~$M_2$ shows that \eqref{gluing_gf} can be written also in other forms. Recall that even for gluing one tetrahedron to the boundary in the way corresponding to a Pachner move $3\to 1$, we could write formula~\eqref{31} in three different ways.
\end{remark}

\begin{remark}
\label{rem:zero}
The general case of gluing several manifolds by some of their boundary components is reducible to a chain of the following two operations:
\begin{itemize}
	\item gluing two connected manifolds over one boundary component and
	\item gluing two identical but oppositely oriented boundary components of one connected manifold.
\end{itemize}
In this section, we have considered the first operation. As for the second one, the most straightforward approach to it gives identical zero for the generating function of the result of gluing, in the same way as in the ``Euclidean'' case, see~\cite[section~4]{tqft2}. The problem of defining the generating function for such cases in a less trivial way appears to be related with the problem of the invariant for $\Sigma \times S^1$, where $\Sigma$ is a closed surface, and $S^1$ --- a circle, for one approach to it see~\cite[Lemma~3]{tqft2}.
\end{remark}

\section{Boundary components of genus zero and connected sums of manifolds}
\label{sec:connected_sums}

We are going to investigate how our generating functions behave when a connected sum is taken. To make a connected sum of two manifolds, one has first to remove the interior of a ball within each of them, and then glue together the spheres --- boundaries of these balls. As we have studied in section~\ref{sec:gluing} what happens under the gluing, it remains to study what happens when we remove the interior of a ball. It is natural to represent this ball as one of the triagulation tetrahedra.

\begin{lemma}
\label{lemma:t_away}
The generating function of invariants for manifold~$M$ without the interior of one (inner) tetrahedron~$a=ijkl$ --- we call the thus obtained manifold~$M'$ --- is
\begin{equation}
\mathbf I_{M'} = \mathbf I_M \, \mathbf I_a ,
\label{t_away}
\end{equation}
where $\mathbf I_a$ is the generating function\footnote{Recall that $\mathbf I_a$ is, according to remark~\ref{rem:tetrahedron}, \emph{one-half} of the tetrahedron function~\eqref{fbf}.} for tetrahedron~$a$ considered as a manifold with boundary.
\end{lemma}

\begin{proof}
First, we prove that the generating function for~$M'$ is of degree one in the anticommuting variables at the edges of~$a$. Stepping away for a moment from the agreement in the beginning of section~\ref{sec:complexes} that the number of vertices in each boundary component should be~$\ge 4$, we can regard the surface of tetrahedron~$a$ as obtained from just two triangles~$ijk$ (with identified edges of the same names) by a two-dimensional Pachner move $1\to 3$. It follows then from lemma~\ref{lemma:13} that $\mathbf I_{M'}$ has degree one in the totality of Grassmann generators $a_{il}$, $a_{jl}$ and~$a_{kl}$ and, moreover, the coefficients at these three generators differ only in nonvanishing numerical factors --- namely, $\zeta_{ij}\zeta_{kl}$, $\zeta_{ik}\zeta_{lj}$ and~$\zeta_{il}\zeta_{jk}$ respectively.

As all the vertices $i,j,k,l$ are here on the equal footing, it follows easily that $\mathbf I_{M'}$ has in fact degree one in the totality of all Grassmann generators for the six edges of~$a$, that the coefficients at these generators are proportional to those in the tetrahedron function~\eqref{fbf}, and there cannot be any term in~$\mathbf I_{M'}$ containing no Grassmann generators corresponding to edges of~$a$. This means that
$$
\mathbf I_{M'} = \mathbf F \, \mathbf I_a
$$
for some function~$\mathbf F$ of Grassmann generators living on other (than the surface of~$a$) components of~$\partial M'$.

To find~$\mathbf F$, we glue back tetrahedron~$a$ to~$M'$ and use formula~\eqref{gluing_gf}, which immediately gives $\mathbf F=\mathbf I_M$.
\end{proof}

\begin{theorem}
\label{th:cs}
The generating function of invariants of a connected sum $M = M_1 \# M_2$ of manifolds is the product of generating functions for $M_1$ and~$M_2$. 
\end{theorem}

\begin{proof}
We take one inner tetrahedron in the triangulation of~$M_1$ and one inner tetrahedron in the triangulation of~$M_2$, remove their interiors and glue together their boundaries. Then we use lemma~\ref{lemma:t_away} and formula~\eqref{gluing_gf}.
\end{proof}

\section{Examples of calculations}
\label{sec:examples}

\subsection{Sphere $S^3$}

According to what we have already said in remark~\ref{rem:1/2},
$$
\mathbf I_{S^3} = 1.
$$

\subsection{Solid torus}
\label{subsec:st}

We consider a solid torus with the boundary triangulation whose development is shown in figure~\ref{fig:st}.
\begin{figure}[ht]
\centering
\includegraphics[scale=0.25]{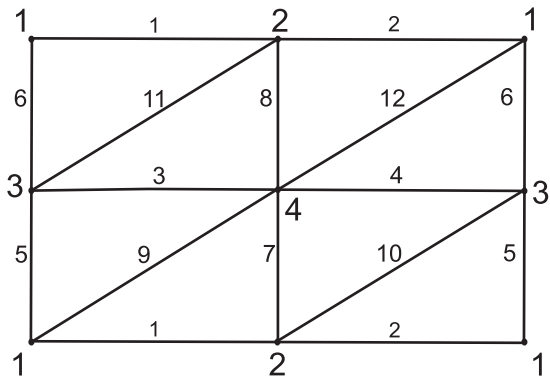}
\caption{Development of the triangulation of a solid torus}
\label{fig:st}
\end{figure}
In it, bigger numbers correspond to vertices, while smaller numbers denote edges and serve as subscripts at the corresponding anticommuting variables. The meridian of the torus goes along edges $5$ and~$6$ (or $7$ and~$8$).

The generating function can be calculated, e.g., using the triangulation of the solid torus of six tetrahedra described in~\cite[Subsection~6.1]{tqft} and using formula~\eqref{s} and lemma~\ref{lemma:t_i}. The answer can be written as
\begin{multline}
\mathbf I_{\textrm{solid torus}} = \frac{1}{2} \, \zeta_{13}^2 \zeta_{24}^2 (a_5-a_6) (a_7-a_8) \bigl( \zeta_{12}\zeta_{34}(a_1+a_3) - \zeta_{13}\zeta_{24}(a_5+a_7) \\ + \zeta_{14}\zeta_{23}(a_9+a_{11}) \bigr) \bigl( \zeta_{12}\zeta_{34}(a_2+a_4) - \zeta_{13}\zeta_{24}(a_5+a_7) + \zeta_{14}\zeta_{23}(a_{10}+a_{12}) \bigr) .
\label{st}
\end{multline}

The function~\eqref{st} is, for instance, efficient enough as to detect the meridians of the torus: if we substitute in~\eqref{st} either $a_6=a_5$ or $a_8=a_7$, it turns into zero, but this by no means happens if we put, say, $a_{12}=a_9$ or $a_2=a_1$. This is due to factors $(a_5-a_6)$ and~$(a_7-a_8)$ in~\eqref{st}, and the following lemma shows that they are not accidental.

\begin{lemma}
\label{lemma:thin}
If~$\partial M$ has exactly one connected component\footnote{most likely, this condition is superfluous for the lemma}, a triangulation of~$\partial M$ is such that there are two edges $p$ and~$q$ forming a circle, and this circle is contractible into a point within~$M$, then the factor $(a_p-a_q)$ can be singled out in~$\mathbf I_M$.
\end{lemma}

\begin{proof}
Contract the circle of edges $p$ and~$q$ into the single edge~$p$. Manifold~$M$ will thus become singular in the neighborhood of~$p$; nevertheless, we can consider its state sum~\eqref{s} for this singular manifold~$M'$. To return back to~$M$, we can glue to~$M'$ two tetrahedra $a$ and~$b$ in such way that $a$~is glued by two of its faces to two triangles adjoining~$p$, while~$b$ --- to the two remaining faces of~$a$.

If now we calculate first the state sum just for the two tetrahedra $a$ and~$b$ glued together this way, we find that it is $\zeta_{ij}^2 (a_p-a_q)$, where $i$ and~$j$ are the ends of both $p$ and~$q$. To finish the proof, it remains to use lemma~\ref{lemma:t_i}.
\end{proof}

\subsection{Solid pretzel}
\label{subsec:pretzel}

Solid pretzel can be obtained, for instance, by gluing two solid tori of subsection~\ref{subsec:st} over one boundary triangle. Thus, the state sum for the solid pretzel is just the product of two state sums for tori --- \eqref{st} without the factor~$1/2$, with the three Grassmann variables at the edges forming the boundary of the mentioned triangle identified.

One can check, in the same way as in subsection~\ref{subsec:st}, that this state sum is also efficient enough to distinguish between the contractible circles in the boundary of solid torus and, say, its parallels (the parallels of the glued tori).

\subsection{$S^3$ without tubular neighborhoods of two unknots: unlinked and linked}

In the case of two unlinked unknots, this manifold is homeomorphic to the connected sum of two solid tori. Its generating function of invariants is, according to theorem~\ref{th:cs}, the product of two expressions~\eqref{st} for tori, but this time with no identification of variables. We can thus again single out the meridians of the mentioned tori in the same way as in subsections \ref{subsec:st} and~\ref{subsec:pretzel}.

On the other hand, $S^3$ without tubular neighborhoods of two \emph{linked} unknots is homeomorphic to $T^2 \times I$, where $T^2$ is the two-dimensional torus, and $I=[0,1]$. In this case, obviously, no special ``meridian'' can be indicated in any way. In particular, this is reflected in our generating function, which is thus different from the case of unlinked unknots. We do not write out here the quite cumbersome expression for this function.

\subsection{Lens spaces without tubular neighborhoods of unknots}

There exist also very interesting manifolds with toric boundary --- lens spaces without tubular neighborhoods of unknots --- where we were able to calculate at least some invariants --- components of our generating function. The results look very nontrivial and need further investigation. We refer the reader to~\cite[Subsection~6.2]{kkm} for some explicit formulas.

\section{Discussion}
\label{sec:discussion}

\subsection{Renormalization and chain complexes}

As we noted in subsection~\ref{subsec:renorm}, the ``na\"{\i}ve'' state-sum invariant~\eqref{s} turns in many cases into zero --- in other words, becomes infinitely small --- and needs a renormalization. It this paper, we performed this renormalization by means of introducing new variables, united into an algebraic (acyclic in many cases) complex. In physics, such new variables may correspond to new physical entities.

An interesting question is: can algebraic complexes be of use in other cases when a renormalization is needed in a physical theory?

\subsection{Less simple models}

What we have considered in this paper is a ``scalar model'' in the sense that scalar --- complex --- quantities were assigned to tetrahedra and vertices. There exist, however, models where elements of an associative algebra, e.g., matrices, play similar roles. Our next aim is to investigate such models, which can be called, due to the noncommutativity of matrix algebras, ``more quantum'' than the one considered in this paper.

One more intriguing area is to study such models over finite fields.

\subsection{Multidimensional generalizations}

An attractive feature of our theory is that it is not limited to three-dimensional manifolds. For instance, the generalization of (a solution to) pentagon equation onto four dimensions must correspond to the Pachner move $3\to 3$, and it does not make much difficulty to write such algebraic relations, again it terms of anticommuting variables, starting, e.g., from formulas already written in \cite{33} or~\cite{kiev}. We plan to present many such relations in our further works.


\begin {thebibliography}{99}

\bibitem{atiyah}
M. Atiyah, Topological quantum field theories, Publications Math\'ematiques de l'IH\'ES, \textbf{68} (1988), 175--186.

\bibitem{atiyah-book}
M. Atiyah, The geometry and physics of knots, Cambridge Univ. Press, Cambridge (1990).

\bibitem{B} F.A.~Berezin, Introduction to superanalysis. Mathematical Physics and Applied Mathematics, vol.~9, D.~Reidel Publishing Company, Dordrecht, 1987.

\bibitem{dkm}
J. Dubois, I.G. Korepanov, E.V. Martyushev, A finite-dimensional TQFT: invariant of framed knots in manifolds, arXiv:math/0605164v3 (2009).

\bibitem{kkm} 
R.M. Kashaev, I.G. Korepanov, E.V. Martyushev, A finite-dimensional TQFT for three-manifolds based on group $\mathrm{PSL}(2,\mathbb C)$ and cross-ratios, arXiv:0809.4239v1 (2008).

\bibitem{33}
I.G. Korepanov, Euclidean 4-simplices and invariants of four-dimensional manifolds: I.~Moves $3\to 3$, Theor. Math. Phys. \textbf{131}
(2002), 765--774.

\bibitem{kiev}
I.G. Korepanov, Pachner move $3\to 3$ and affine volume-preserving geometry in $\mathbb R^3$, SIGMA, vol.~1 (2005), paper~021, 7~pages.

\bibitem{tqft}
I.G. Korepanov, Geometric torsions and invariants of manifolds with a triangulated boundary, Theor. Math. Phys., vol.~158 (2009), 82--95.

\bibitem{tqft2}
I.G. Korepanov, Geometric torsions and an Atiyah-style topological field theory, Theor. Math. Phys., vol.~158 (2009), 344--354.

\bibitem{pachner_moves}
W.B.R. Lickorish, Simplicial moves on complexes and manifolds, Geom. Topol. Monographs \textbf{2} (1999), 299--320.

\bibitem{PL_and_Top}
E. Moise, Affine structures in 3-manifolds, V, Ann. of Math., \textbf{56} (1952), 96--114.

\bibitem{turaev}
V.G. Turaev, Introduction to combinatorial torsions, Boston: Birkh\"auser, 2000.

\end{thebibliography}
\end{document}